\documentclass[11pt,reqno]{amsart}
\usepackage{amsmath,amsxtra,amssymb,amsthm,amsfonts,bm}

\usepackage[latin1]{inputenc}
\usepackage[cal=cm]{mathalfa}
  \usepackage[foot]{amsaddr}

\usepackage{bbold}
\usepackage{bbm}
\usepackage{nonfloat}
\usepackage{braket}
\usepackage{dsfont}
\usepackage{mathdots}
\usepackage{mathtools}
\usepackage{enumerate}
\usepackage{csquotes}
\usepackage{stmaryrd}
\usepackage{graphicx}
\usepackage{stackengine}
\usepackage{scalerel}
\usepackage{array}
\usepackage{makecell}
\newcolumntype{x}[1]{>{\centering\arraybackslash}p{#1}}
\usepackage{tikz}
\usetikzlibrary{shapes.geometric, shapes.misc, positioning, arrows, decorations.pathreplacing, angles, quotes}
\usepackage{booktabs}
\usepackage{xfrac}
\usepackage{siunitx}
\usepackage{centernot}
\usepackage{comment}
\usepackage{chngcntr}

\newtheorem{theorem}{Theorem}
\numberwithin{theorem}{section} 

\newtheorem*{theorem*}{Theorem}
\newtheorem{proposition}[theorem]{Proposition}
\newtheorem*{proposition*}{Proposition}
\newtheorem{lemma}[theorem]{Lemma}
\newtheorem*{lemma*}{Lemma}
\newtheorem{corollary}[theorem]{Corollary}
\newtheorem*{cor*}{Corollary}

\newtheorem*{cj*}{Conjecture}
\newtheorem{definition}[theorem]{Definition}
\newtheorem*{Def*}{Definition}

\makeatletter
\def\thmhead@plain#1#2#3{%
  \thmname{#1}\thmnumber{\@ifnotempty{#1}{ }\@upn{#2}}%
  \thmnote{ {\the\thm@notefont#3}}}
\let\thmhead\thmhead@plain
\makeatother

\theoremstyle{definition}
\newtheorem{rem}[theorem]{Remark}

\newtheorem{example}[theorem]{Example}

\newcommand{\vertiii}[1]{{\left\vert\kern-0.25ex\left\vert\kern-0.25ex\left\vert #1 
    \right\vert\kern-0.25ex\right\vert\kern-0.25ex\right\vert}}

\newcommand{\ten}{\otimes}

\newcommand{\pl}{\hspace{.1cm}}

\newcommand{\ran}{\rangle}
\newcommand{\lan}{\langle}
\newcommand{\al}{\alpha}
\renewcommand{\si}{\sigma}

\newcommand{\la}{\lambda}

\newcommand{\bb}{\begin{equation}}
\newcommand{\bbb}{\begin{equation*}}
\newcommand{\ee}{\end{equation}}
\newcommand{\eee}{\end{equation*}}

\newcommand{\B}{{\mathcal{B}}}

\newcommand{\cT}{\mathcal{T}}

\newcommand{\norm}[2]{\|  #1  \|_{#2}}
\newcommand{\id}{\operatorname{id}}
\newcommand{\tr}{\operatorname{Tr}}

\newcommand{\cD}{\mathcal{D}}
\newcommand{\cN}{\mathcal{N}}
\newcommand{\cM}{\mathcal{M}}

\newcommand{\cH}{\mathcal{H}}
\newcommand{\cX}{\mathcal{X}}
\newcommand{\cB}{\mathcal{B}}

\newcommand{\cL}{\mathcal{L}}
\newcommand{\cP}{\mathcal{P}}
\newcommand{\cE}{\mathcal{E}}

\newcommand{\cA}{\mathcal{A}}

\DeclareMathOperator{\Tr}{Tr}

\DeclareMathAlphabet{\pazocal}{OMS}{zplm}{m}{n}

\DeclareMathOperator{\supp}{supp}

\DeclareMathOperator{\dom}{dom}

\newcommand{\cK}{\mathcal{K}}

\newcommand{\D}{\pazocal{D}}

\newcommand{\lsmatrix}{\left(\begin{smallmatrix}}
\newcommand{\rsmatrix}{\end{smallmatrix}\right)}

\stackMath

\stackMath

\makeatletter
\newcommand*\rel@kern[1]{\kern#1\dimexpr\macc@kerna}
\newcommand*\widebar[1]{%
  \begingroup
  \def\mathaccent##1##2{%
    \rel@kern{0.8}%
    \overline{\rel@kern{-0.8}\macc@nucleus\rel@kern{0.2}}%
    \rel@kern{-0.2}%
  }%
  \macc@depth\@ne
  \let\math@bgroup\@empty \let\math@egroup\macc@set@skewchar
  \mathsurround\z@ \frozen@everymath{\mathgroup\macc@group\relax}%
  \macc@set@skewchar\relax
  \let\mathaccentV\macc@nested@a
  \macc@nested@a\relax111{#1}%
  \endgroup
}

\usepackage[pdftex]{hyperref}
\hypersetup{colorlinks=true,linkcolor=blue,filecolor=magenta,urlcolor=blue}
\usepackage{cleveref}

\title[Ricci curvature of quantum channels]{Ricci curvature of quantum channels \\on non-commutative transportation metric spaces}

\author{Li Gao}
 \email[Li Gao]{li.gao@tum.de}
\author{Cambyse Rouz\'{e}}
 \address{Zentrum Mathematik, Technische Universit\"{a}t M\"{u}nchen, 85748 Garching, Germany}
\email[Cambyse Rouz\'{e}]{cambyse.rouze@tum.de}

\begin{document}
\begin{abstract}
Following Ollivier's work \cite{Ollivier2009}, we introduce the coarse Ricci curvature of a quantum channel as the contraction of  non-commutative metrics on the state space. These metrics are defined as a non-commutative transportation cost in the spirit of \cite{Guillin2008,Gozlan2006}, which gives a unified approach to different quantum Wasserstein distances in the literature. We prove that the coarse Ricci curvature lower bound and its dual gradient estimate, under suitable assumptions, imply the Poincar\'e inequality (spectral gap) as well as transportation cost inequalities. Using intertwining relations, we obtain positive bounds on the coarse Ricci curvature of Gibbs samplers, Bosonic and Fermionic beam-splitters as well as Pauli channels on $n$-qubits. 
\end{abstract}
\maketitle

\section{Introduction}
The Ricci curvature of a Riemannian manifold $M$ quantifies the amount by which the volume of small geodesic balls in $M$ deviates from that in the standard Euclidean space. When being globally strictly positive, the Ricci curvature governs many geometric and analytic properties of $M$ such as its diameter via the Bonnet-Meyers theorem, isoperimetric inequalities and the concentration of measure on $M$ (see the monograph \cite{villani2009optimal}). In their seminal work \cite{bakry1984hypercontractivite,Bakry1985}, Bakry and \'{E}mery established an inequality for diffusive Markov processes, known as the curvature dimension condition, which captures the features of a uniform Ricci curvature lower bound. With positive curvature bound, this inequality implies well-known functional inequalities, such as Sobolev, logarithmic Sobolev and Poincar\'{e} inequalities (see e.g. \cite{SaloffCoste1994,bakry1994hypercontractivite}). These inequalities are in turn directly applicable to the analysis of the mixing time of the diffusion process as well as to derive concentration inequalities for their invariant measures (see \cite{Ledoux2005,Chafai,Bakry2014} for more details on these topics).

In the past decades, the notion of Ricci curvature has been largely extended beyond Riemannain manifolds. Interestingly, when considering Markov processes on a discrete metric space (which are generically non-diffusive),  various notions of Ricci curvature bound were proposed. Among which are the Bakry-\'{E}mery curvature dimension condition via gradient form \cite{Bakry1985}, Lott-Sturm-Villani's synthetic theory of Ricci curvature on metric measure spaces \cite{Lott2009,Sturm2006}, Erbar and Maas' entropic Ricci curvature for finite Markov chains \cite{Maas2011,Erbar2012} and Ollivier's coarse Ricci curvature for discrete-time Markov chains on metric spaces \cite{Ollivier2009}.
All the above notions reduce to the original definition in the Riemannian setting, while the last three of them substantially rely on the theory of optimal transport \cite{villani2009optimal}. 

Let $(\cX,d)$ be a complete separable metric space with distance function $d$. For $1\le p<\infty$, the Wasserstein $p$-distance between two probability measures $\mu,\nu$ is defined as
$$W_p(\mu,\nu):=\inf_{\pi\in\Pi(\mu,\nu)}\,\Big(\int_{\cX\times \cX}\,d(x,y)^p\,\pi(dx,dy)\Big)^{\frac{1}{p}}\, ,$$ 
where the infimum is taken over all joint distributions $\pi$ on $\cX\times \cX$ whose marginals coincide with $\mu$ and $\nu$. 
For $\cX=M$ being a compact Riemannian manifold, the Ricci curvature is uniformly lower bounded by a constant $\la\in\mathbb{R}$ if and only if for any two probability measures $\mu,\nu$ and all $t\ge 0$,
\begin{align}
    W_p(\mu\circ e^{t\Delta},\nu \circ e^{t\Delta})\le e^{-t\la} \,W_p(\mu,\nu)
\end{align}
for either $p=1$ or $2$. Here $e^{t\Delta}$ denotes the heat semigroup generated by the Laplace-Beltrami operator $\Delta$. Motivated by the above formulation, Ollivier \cite{Ollivier2009} introduced the notion of coarse Ricci curvature bound $\kappa\in \mathbb{R}$ of a Markov chain $\cP:L_\infty(\cX)\to L_\infty(\cX)$ on a generic Polish space $(\cX,d)$ as follows: for any two probability measures $\mu,\nu$,
\begin{align}\label{Olliviercorasecurvature}
    W_1(\mu\circ \cP,\nu\circ \cP)\le \,(1-\kappa)\,W_1(\mu,\nu)\,.
\end{align}
In particular, Ollivier's definition recovers the Riemannian setting for $\cP=e^{t\cL}$ and $\kappa(t)=1-e^{-\la\,t}$. By Monge-Kantorovich duality, \eqref{Olliviercorasecurvature} is also equivalent to the contraction (when $\kappa>0$) of Lipschitz constants:
\begin{align}\label{gradientestimate}
   \|P(f)\|_{\operatorname{Lip}}\,\le\,(1-\kappa)\,\|f\|_{\operatorname{Lip}}\,.
\end{align}

On the other hand, Erbar and Maas \cite{Erbar2012} introduced the entropic Ricci curvature bound for finite Markov chains based on Benamou-Brenier's dynamical formulation of the $W_2$ distance \cite{Benamou2000} (see also \cite{Otto2000,Ambrosio2012}). Their approach was inspired by Otto's observation that the heat semigroup on $M$ is the gradient flow of the entropy with respect to the Wasserstein $2$-distance. Given a continuous time Markov chain $e^{t\cL}$ on a finite sample space $\cX$, Erbar and Maas identified the Wasserstein distance $W_{2,\cL}$ (depending on the generator $\cL$) such that the semigroup $e^{t\cL}$ is the gradient flow of the relative entropy functional. With this definition, they introduced the entropic Ricci curvature bound $C\in\mathbb{R}$, which can be equivalently formulated in terms of an exponential contraction of $W_{2,\cL}$,
\begin{align}
    W_{2,\cL}(\mu\circ e^{t\cL},\nu\circ e^{t\cL})\le e^{-t C}\,W_{2,\cL}(\mu,\nu)\,.
\end{align}
The framework of Erbar and Maas was extended by Carlen and Maas \cite{carlen2014analog,carlen2017gradient,Carlen2019} to the non-commutative setting of quantum Markov semigroups on matrix algebras (see also \cite{datta2020relating,Mittnenzweig2017}), and later generalized to symmetric semigroups on finite von Neumann algebras \cite{wirth2018noncommutative,li2020graph}. 
The entropic curvature bound is known to imply exponential convergence in terms of the relative entropy as well as other functional inequalities such as the Poincar\'e and  transportation cost inequalities, in both discrete classical and quantum settings.
It remains an open problem to compare Erbar and Maas' entropic Ricci curvature with Ollivier's coarse Ricci curvature for finite Markov chains. For instance, the application to entropic convergence remains open for Ollivier's coarse curvature \cite{Eldan2017}.

In recent years, various works have studied Wasserstein distances on quantum systems \cite{carlen2017gradient,DePalma2021,de2021quantum,caglioti2021towards}.
 While the entropic Ricci curvature has found many applications in the study of quantum Markov semigroups (see e.g. \cite{datta2020relating,brannan2020complete}),
 Ollivier's notion of coarse Ricci curvature was left largely unexplored in the non-commutative realm. The motivation of the present paper is to fill this gap and study the curvature condition for quantum channels via a unified approach.  

Let us illustrate our idea from the perspective of quantum metric spaces. These spaces originate from Connes' work on non-commutative geometry \cite{connes1990geometrie} and were later studied by Rieffel \cite{rieffel2004compact} and others for their geometric properties such as the 
Gromov-Hausdorff convergence \cite{junge2018harmonic}. 
A ($W^*$-)quantum metric space $(\cM,L)$ is given by a von Neumann algebra equipped with a semi-norm $L:\cA\to [0,\infty)$ defined on a dense subalgebra $\cA\subset \cM$. This semi-norm often arises from a (non-commutative) differential structure,
and can be viewed as an abstraction of the notion of Lipschitz constant. It induces the following Connes distance on states: for any two states $\omega_1,\omega_2$,
\[ \cT_L(\omega_1,\omega_2):=\sup\{\omega_1(x)-\omega_2(x)|\,L(x)\le 1\}\pl. \]
Examples include different quantum Wasserstein $1$-distances considered in the literature \cite{DePalma2021,datta2020relating,Junge2014} .  
In analogy to the classical picture \eqref{Olliviercorasecurvature} and \eqref{gradientestimate}, the contraction of the distance $\cT_L$ can be derived by the contraction of the semi-norm,
\begin{align}\label{gradientestimate2}
    L(\cP(x))\le\,(1-\kappa)\,L(x)\,,
\end{align}
as a dual estimate. 
More generally, following the abstract setting of \cite{Gozlan2006,Guillin2008}, a non-commutative transportation cost $\cT_\cB:\cD(\cM)\times\cD(\cM)\to[0,\infty]$ can be defined as follows
\begin{align}
    \mathcal{T}_\cB(\omega_1,\omega_2):=\sup_{(x,y)\in\cB}\,\omega_1(x)-\omega_2(y)\,.
\end{align}
Here, the set $\cB \subset\cM_\mathbb{R}\times \cM_\mathbb{R}$ is a binary relation on the real part $\cM_{\mathbb{R}}$ satisfying certain axioms (specified in Section 3). 
The Connes distance $\cT_L$ can then be recovered as a special case after taking $\cB=\{(x,x)|\,x\in \cM_{\mathbb{R}},\,L(x)\le 1\}$. Interestingly,
due to a recent result by Wirth \cite{wirth2021dual},
the quantum Wasserstein $2$-distance of Carlen and Maas also fits into this framework. This allows us to take a unified approach to Ricci curvature on both quantum Wasserstein $1$ and $2$-distances:
 we say that a triple $(\cM,\cB,\cP)$ satisfies a coarse Ricci curvature lower bound $\kappa\in\mathbb{R}$ if for any two states $\omega_1,\omega_2$,
\begin{align}\label{curvaturegeneral}
    \cT_\cB(\omega_1\circ\cP,\omega_2\circ\cP)\le\,(1-\kappa)\,\cT_\cB(\omega_1,\omega_2)\,.
\end{align}
The rest of this paper is devoted to an analysis of the coarse Ricci curvature \eqref{curvaturegeneral} as well as the dual estimate \eqref{gradientestimate2} and is organize as follows: In Section \ref{sec:preliminary}, we review some basic facts on von Neumann algebras, derivations and operator means. We work in this general setting in order to include both classical and quantum examples in possibly infinite dimensions. Section \ref{sec:NCTC} introduces the non-commutative transportation cost and reviews various examples in the literature. This includes the De Palma et al's Quantum $1$-Wasserstein distance \cite{de2021quantum} and Carlen and Mass' Quantum transport metric \cite{carlen2017gradient,Carlen2019}. Section \ref{sec:curvature} starts with some basic properties of coarse Ricci curvature of a triple $(\cM,\cB,\cP)$ and its dual Lipschitz/gradient estimate. We show that a positive coarse Ricci curvature bound has applications to derive diameter estimates, Poincar\'e and transportation cost inequalities. 
In Sections \ref{sec:intertwining} and \ref{sec:transference}, we provide two approaches to derive the coarse Ricci curvature bound. One is the intertwining relation with derivation, which is a primary source of examples in the continuous time setting. The other is the group transference method, which allows us to pass Lipschitz/gradient estimates from classical Markov maps on a group to quantum channels via group representation.
Finally, Section \ref{examples} provides some new examples of positive coarse Ricci curvature, including Gibbs samplers, Bosonic channels and Pauli channels on $n$-qubits.

\section{Preliminary}\label{sec:preliminary}
\subsection{States, channels and Markov semigroups}
We denote by $\cB(\cH)$ the bounded operators on a Hilbert space $\cH$. Recall that a von Neumann algebra $\cM$ is a unital weak$^*$-closed subalgebra of some $\cB(\cH)$. Throughout the paper, we write $1$ as the identity operator in $\cM$, and $\id$ as the identity map between two algebras.
We denote $\cM_{\mathbb{R}}:=\{x\in\cM|\,x=x^*\}$ for the real part of $\cM$ and $\cM_+=\{a^*a | a\in \cM\}$ for the positive cone. A linear functional $\phi:\cM\to \mathbb{C}$ is a state if $\phi(1)=1$ and $\phi(x)\ge 0$ for all $x\in \cM_+$; $\phi$ is normal if $\phi$ is weak$^*$-continuous. We denote by $\cD(\cM)$ the set of normal states on $\cM$.

Throughout the paper, we will mostly consider a semi-finite von Neumann algebra $\cM$. That is, 
 $\cM$ admits a normal faithful semi-finite trace $\tau:\cM_+\to [0,\infty]$. More precisely, $\tau$ satisfies
\begin{enumerate}
\item[(i)] if $\tau(x^*x)=0$, then $x=0$;
\item[(ii)] for an increasing net $\{x_{\al}\}\subset \cM_+$,  $\sup\tau(x_\al)=\tau(\sup x_{\al})$;
\item[(iii)] for any nonzero $x\ge 0$, there exists a nonzero $y\le x$ such that $\tau(y)<\infty$.
\end{enumerate}
For example, $\cB(\cH)$ equipped with the standard trace $\tr$ is a semi-finite von Neumanna algebra. $\cM$ is a finite von Neumann algebra if
in additional $\tau(1)<\infty$. All finite dimensional $\cM$ are finite.

Denote $\cM_0=\bigcup_{\tau(e)<\infty}e\cM e$ where the union is over all finite projection $e\in \cM$ with $\tau(e)<\infty$. For $1\le p<\infty$, the $L_p$ space $L_p(\cM)$ is defined as the completion of $\cM_0$ with respect to the norm
\[\norm{a}{L_p(\tau)}=\tau(|a|^{p})^{1/p}\pl.\]
We will often use the short notation $\norm{\cdot}{p}:=\norm{a}{L_p(\tau)}$ if no confusion can subsist. We identify $L_\infty(\cM):=\cM$ and the predual space $\cM_*\cong L_1(\cM)$ via the duality
\[a\in L_1(\cM)\longleftrightarrow \phi_a\in \cM_*,\pl  \qquad \phi_a(x)=\tau(ax)\pl.\]
We say $\rho \in L_1(\cM)$ is a density operator (or simply a density) if $\rho\ge 0$ and $\tau(\rho)=1$. The set of all densities corresponds to $\cD(\cM)$, the normal states of $\cM$. 

A quantum channel $\cP:\cM\to\cM$ is a normal, unital, and completely positive (UCP) map. Its pre-adjoint $\cP^\dagger:\cM_*\to \cM_*$ is  completely positive trace preserving (CPTP), and hence is a map from $\cD(\cM)$ to itself. A quantum Markov semigroup is a family of linear maps $\displaystyle t\mapsto \cP_t:\cM\to \cM$ with the following properties:
\begin{enumerate}
\item[(i)] $\cP_t$ is a normal UCP map (quantum channel) for all $t\ge 0$;
\item[(ii)] $\cP_t\circ \cP_s=\cP_{s+t}$ for any $t,s\ge 0$ and $\cP_0=\id$;
\item[(iii)] for each $x\in \cM$, $t\mapsto \cP_t(x)$ is continuous in the weak$^*$-topology.
\end{enumerate}
The generator of the semigroup is defined as
\[\pl \cL (x)=w^*-\lim_{t\to 0}\, \frac{\cP_t(x)-x}{t}\pl,\quad \pl \cP_t=e^{t\,\cL }\pl,\]
where $\cL$ is a closable densely defined operator on $L_\infty(\cM)$. We say that a quantum Markov semigroup $\cP_t$ is GNS-symmetric with respect to a state $\omega$ if
\[ \omega(x^*\cP_t(y))=\omega(\cP_t(x)^*y)\pl , \pl \quad \forall x,y\in \cM.\]
In particular, $\cP_t$ is symmetric if
$\cP_t$ is GNS-symmetric to the trace $\tau$.

Recall that the relative entropy of two states $\rho$ and $\si$ is given by
\[D(\rho\|\si)=\tau(\rho \log \rho-\rho\log \si)\pl,\]
provided $\rho \log \rho, \rho\log \si\in L_1(\cM)$. 
The relative entropy $D(\rho\|\si)$ measures how well the state $\rho$ can be distinguished from $\sigma$ by quantum measurements \cite{blahut1974hypothesis,hiai1991proper,ogawa2005strong}. In particular, it compares to the $L_1$-norm by the Pinsker inequality (see \cite[Theorem 5.5]{ohya2004quantum})
\[ \frac{1}{2}\norm{\rho-\si}{1}^2\,\le \,D(\rho\|\si)\pl.\]
 Let $\cN\subset \cM$ be a von Neumann subalgebra. A conditional expectation $E:\cM\to \cN$ is a completely positive unital map such that $E(axb)=aE(x)b$ for any $a,b\in \cN$ and $x\in \cM$. We say that $E:\cM\to \cN$ preserves the state $\si$ if $\si\circ E=\si$. For any $E$-invariant state $\si$, we have the chain rule 
 \cite[Theorem 5.15]{ohya2004quantum}
\begin{align} D(\rho\|\si)=D(\rho\|\rho\circ E)+D( \rho\circ E\|\si)\pl.\label{eq:chain}\end{align}

\subsection{First order calculus}\label{sec:differential}
One natural and primary source of non-commutative metric is from the so-called first order differential structure introduced in \cite{cipriani2003derivations}.
We say that a tuple $(\cA,\cH,l,r, \partial)$ is a \emph{first order differential structure} on $\cM$  if it consists of a weak$^*$ dense subalgebra $\cA\subset\cM$,
a Hilbert space $\cH$, commuting non-degenerate $*$-homomorphisms $l:\cM\to \cB(\cH),r:\cM^{\operatorname{op}}\to \cB(\cH)$ 
and a closed derivation $\partial: \cA\to \cH$ satisfying the Leibiniz rule
\[\partial(xy)=l(x)\partial y+r(y)\partial x\pl, \pl  \forall \pl x,y\in\cA\pl. \]
Here $\cM^{\operatorname{op}}$ is the opposite algebra of $\cM$ defined as
\[\cM^{\operatorname{op}}=\{a| a\in \cM\}\pl\pl, \pl \pl a\cdot b= ba.\]

Let $\cM$ be a finite von Neumann algebra equipped with trace $\tau$ and let $\cP_t=e^{t\,\cL }:\cM\to \cM$ be a $\tau$-symmetric quantum Markov semigroup with the generator $\cL$. We denote the Dirichlet subalgebra as $\cA:=\cM\cap\operatorname{dom}((-\cL)^{\frac{1}{2}})$.
Recall that
the gradient form or \emph{Carr\'e du Champ} operator is defined for all $x\in\cM\cap\operatorname{dom}(\cL)$ (and for all $x\in\cA$ after extension) as:
\begin{align}\label{gradientform}
  2\,  \Gamma(x,x):=\cL(x^*x)-\cL(x^*)x-x^*\cL(x)\,.
\end{align}
It is proved by Cipriani and Sauvageot \cite{cipriani2003derivations}
that if the gradient form $\Gamma$ is regular, i.e. $\Gamma(x,x)\in L_1(\cM)$ for all $x\in \cA$, then $\Gamma$ can be realized from a (unique) first order differential structure $(\cA,\cH, l,r, \partial)$ as follows
\begin{align}\label{eq:derivation} \lan \partial(x),\partial(y)z \ran_\cH=\tau(\Gamma(x,y)z)\pl,\end{align}
In particular, $\cL=-\partial^*\partial$ where $\partial^*$ is the adjoint operator of $\partial$. This differential structure is an essential ingredient used in \cite{wirth2018noncommutative} to study non-commutative Wasserstein metrics induced by symmetric quantum Markov semigroups. It was further proved in the preprint \cite{JRS} that the above Hilbert space $\cH$ can be chosen as the $L_2$-space $L_2(\hat{\cM})$ of a finite von Neumann algebra $(\hat{\cM},\tau)$ such that $\cM\subset \hat{\cM}$ with induced trace and left and right actions defined as
\begin{align}l(x)\xi:=x\xi,\quad  r(x)\xi:=\xi x, \quad \forall \pl x,y\in\cA ,\quad  \xi\in L_2(\hat{\cM})\pl.\label{eq:algebra}\end{align}
In this setting, the gradient form $\Gamma$ can be expressed as
\begin{align}\label{eq:gradientform} \Gamma(x,y)=E_{\cM}(\partial(x)^*\partial(y))\,,\end{align}
where $E_{\cM}:\hat{\cM}\to \cM$ is the trace preserving conditional expectation. We note that similar constructions are also obtained in \cite{carlen2017gradient} for finite dimensional GNS-symmetric semigroup with respect to a faithful state $\sigma$. This constitutes a crucial step in the definition of the quantum Wasserstein 2-distance.

\subsection{Operator mean}
\label{sec:operator}



We briefly review the notion of an operator mean in the sense of Kubo and Ando \cite{Kubo1980}. An operator mean is a binary map $\lambda:\cB(\cH)_+\times \cB(\cH)_+\to \cB(\cH)_+$ satisfying the following properties:
\begin{enumerate}
\item[i)] if $A_1\le A_2$ and $B_1\le B_2$, then $\Lambda(A_1,B_1)\le \Lambda(A_2,B_2)$.
\item[ii)]  $C\Lambda(A,B)C\le\Lambda(CAC,CBC)$ for any $A,B,C\in \B(\cH)_+$
\item[iii)] if $A_n\searrow A$ and $B_n\searrow B$ are decreasing , then $\Lambda(A_n,B_n)\searrow \Lambda(A,B)$.
\item[iv)] $\Lambda(1,1)=1$ ($1$ is the identity operator)\,.
\end{enumerate}
Here $A_n\searrow A$ means $A_n$ is a decreasing sequence and converges to $A$ in strong operator topology. The simplest example is the arithmetic mean $\Lambda_{\operatorname{ari}}(A,B)=\frac{1}{2}(A+B)$. Two trivial examples are the right trivial mean $\Lambda_r(A,B)=B$ and the left trivial mean $\Lambda_l(A,B)=A$. In general, it is proved in \cite[Theorem 3.4]{Kubo1980} that 
every operator mean $\Lambda$ (for invertible $A, B$) can be expressed as
\begin{align}
\Lambda(A,B)=aA+bB+\int_{0}^\infty \frac{1+t}{t}\frac{1}{(tA)^{-1}+B^{-1}} \,d\mu(t)  
\end{align}
where $a=\mu({0}),b= \mu({\infty})$ and $\mu$ is some Radom probability measure on $[0,\infty]$. 

Let $(\cA,\cH,l,r, \partial)$ be a first order differential structure on $\cM$. For a positive $\rho\in \cM_+$, we write
\[\Lambda(\rho):=\Lambda(l(\rho),r(\rho))\]
as a positive operator in $\cB(\cH)_+$. We will be in particular interested in the logarithmic mean
\[\Lambda_{\log}(\rho)=\int_{0}^1l(\rho)^{s}r(\rho)^{1-s}ds \,.\]

\section{Non-commutative transportation metrics}\label{sec:NCTC}
We start by introducing a general notion of non-commutative transportation cost.
Let $\cM$ be a von Neumann algebra. Recall that  $\cM_{\mathbb{R}}:=\{x\in\cM|\,x=x^*\}$ denotes the real part of $\cM$, and $\cD(\cM)$ denotes the set of normal states on $\cM$. Inspired by \cite{Guillin2008,Gozlan2006}, we propose the following definition of a \textit{non-commutative transportation cost}.
\begin{definition}
\label{def:metric}
We say a binary relation $\mathcal{B}\subset \cM_{\mathbb{R}}\times \cM_{\mathbb{R}}$ is a non-commutative transportation metric if it satisfies the following requirements:
\begin{itemize}
    \item[(i)] For all $(x,y)\in \mathcal{B}$, $x\le y$.
    \item[(ii)] For all states $\omega_1,\omega_2\in \cD(\cM)$, there exists $(x,y)\in \cB$ such that $\omega_1(x)-\omega_2(y)\ge 0$.
    \end{itemize}
Given such a binary relation $\mathcal{B}$, we define the \textit{non-commutative transportation cost} $\mathcal{T}_\cB:\cD(\cM)\times \cD(\cM)\to [0,\infty]$ as
\begin{align}
    \mathcal{T}_\cB(\omega_1,\omega_2):=\sup_{(x,y)\in\cB}\,\omega_1(x)-\omega_2(y)\,.
\end{align}
We denote the domain of $\cT_\cB$ as $\operatorname{dom}(\cT_\cB):=\{(\omega_1,\omega_2)\in\cD(\cM)\times \cD(\cM):\,\cT_\cB(\omega_1,\omega_2)<\infty\}$. The map $\cT_\cB$ is clearly jointly convex on its domain.
\end{definition} It is clear that the condition (ii) implies that $\cT_\cB$ is non-negative whereas the condition (i) implies 
$\cT_\cB(\omega,\omega)=0,\,\forall \omega\in\cD(\cM)$. We sometimes also call the pair $(\cM,\cB)$ a non-commutative transportation metric by slight abuse of language. We shall now discuss how this unifies different notions of transportation metrics in the literature.

\subsection{Classical transportation cost}\label{sec:classical}
We briefly recall the abstract notion of commutative transportation cost from \cite{Gozlan2006,Guillin2008}, for later comparison to the non-commutative settings. We refer to Villani's book \cite{villani2009optimal} for more information on this topic.
Let $(\mathcal{X},d)$ be a complete separable metric space (Polish space). The space $L_\infty(\mathcal{X})$ of Borel-measurable bounded functions on $\mathcal{X}$ is a commutative von Neumann algebra. We denote by
$L_\infty(\mathcal{X},\mathbb{R})$ the set of all measurable bounded real functions on $\mathcal{X}$, 
and by $\cD(\mathcal{X})$ the set of all probability measures. The following transportation cost was introduced in \cite{Gozlan2006,Guillin2008}: for $\mu,\nu\in \D(\mathcal{X})$,
\[ \cT_{\cB}(\mu,\nu):=\sup_{(f,g)\in\cB \cap (L_1(\mu)\times L_1(\nu)) } \int f\ d \mu -\int g\ d \nu \pl,\]
where $\mathcal{B}\subset L_\infty(\mathcal{X},\mathbb{R})\times L_\infty(\mathcal{X},\mathbb{R})$ is a binary relation satisfying 
\begin{itemize}
    \item[(i)] For all $(f,g)\in \mathcal{B}$, $f\le g$.
    \item[(ii)] For all $\mu,\nu\in \cD(\mathcal{X})$, there exists $(f,g)\in \cB$ such that $\int f\ d \mu -\int g\ d \nu\ge 0$.
\end{itemize}
The above condition are often satisfied by binary relations induced by a cost function $c:\mathcal{X}\times \mathcal{X} \to\overline{\mathbb{R}}$ as follows
\[\cB_{c}:=\{(f,g)| f(x)-g(y)\le c(x,y)\}\pl.\]
The celebrated Kantorovich duality states that if the cost function $c$ is lower semi-continuous, then
\begin{align} \label{eq:Kantorovich}\cT_{\cB}(\mu,\nu)=&\sup_{(f,g)\in\cB \cap (L_1(\mu)\times L_1(\nu)) } \int f\ d \mu -\int g\ d \nu
\\ =& \inf_{\pi\in \Pi(\mu,\nu) } \int_{\mathcal{X}\times \mathcal{X}}c(x,y)\,d\pi(x,y)\nonumber
\end{align}
where the infimum is over the set $\Pi(\mu,\nu)$ of joint distributions $\pi$ whose marginals are $\pi_1=\mu$ and $\pi_2=\nu$. Based on \eqref{eq:Kantorovich},
the binary relation $\cB_{c}$
satisfies the condition (i) in Definition \ref{def:metric} if $c(x,x)=0$ for all $x\in \mathcal{X}$; $\cB_{c}$
satisfies the condition (ii), or equivalently, the non-negativity of $\cT_{\cB_c}$, if $c(x,y)\ge 0$ for all $x,y\in \mathcal{X}$. Thus a common assumption for the cost function is that
$c:\mathcal{X}\times \mathcal{X} \to [0,\infty]$ is lower semi-continuous and $c(x,x)=0$ for all $x \in \mathcal{X}$. 


Standard examples of transportation costs are Wasserstein distances. The Wasserstein $1$-distance takes the cost function $c(x,y)=d(x,y)$ as the underlying distance on $\cX$, and possesses the following simplified expression as the dual distance to the Lipschitz constant
\begin{align*} W_1(\mu,\nu)=\sup_{f(x)-g(y)\,\le\, d(x,y)} \,\int f\,d\mu - \int g\,d\nu =  \sup_{\norm{\,f\,}{\text{Lip}}\,\le \,1} \int f\,(d\mu -d\nu)\,,
\end{align*}
where we recall that the Lipschitz constant is defined as $\|f\|_{\text{Lip}}:=\sup_{x,y}\frac{|f(x)-f(y)|}{d(x,y)}$. The above simplification follows from the sufficiency of choosing the tight pair $(f,g)$ as 
\[ f(x)=\inf_{x}(g(y)+d(x,y))\pl, \qquad g(y)=\sup_{x}(f(x)-d(x,y)) \,.\]
It is clear that both $\|f\|_{\text{Lip}},\|g\|_{\text{Lip}}\le 1$ by triangle inequality, and it follows that $f=g$.  In particular, this recovers the distance function $d$ by
\[d(x,y)=W_1(\delta_x,\delta_y)\quad  \forall \pl x,y\in \mathcal{X}\,, \]
where $\delta_x$ denotes the Dirac distribution at $x\in \cX$. One can also consider $c_p(x,y)=d(x,y)^p$ for $p\ge 1$, for which the transportation cost is
\begin{align*} \cT_{c_p}(\mu,\nu)=\sup_{f(x)-g(y)\le d(x,y)^p} \,\int f\,d\mu - \int g \,d\nu= \inf_{\pi\in \Pi(\mu,\nu ) } \int_{\mathcal{X}\times \mathcal{X}}d(x,y)^p \,d\pi(x,y)\,.
\end{align*}
For $1<p<\infty$,
\[ W_p(\mu,\nu)=\Big(\inf_{\pi\in \Pi(\mu,\nu) }\, \int_{\mathcal{X}\times \mathcal{X}}d(x,y)^p \,d\pi(x,y)\Big)^\frac{1}{p}\]
are the Wasserstein $p$-distances. The following ``Gluing lemma'' is the key fact behind the proof of the triangle inequality for $W_p$:
\begin{lemma}[(\cite{villani2009optimal}, p. 11)]\label{lemma:glue}
Let $\pi$ be a joint distribution over $\mathcal{X}\times\mathcal{X}$ whose marginal distributions are $\mu_1$ and $\mu_2$, and let $\pi'$ be a joint distribution over $\mathcal{X}\times\mathcal{X}$ whose marginal distributions are $\mu_2$ and $\mu_3$. There exists a joint distribution $\gamma$ over $\mathcal{X}\times\mathcal{X}\times\mathcal{X}$ such that 
\[\gamma_{12}=\pi,\qquad  \gamma_{23}=\pi'\,.\]
\end{lemma}
Based on the above lemma, the triangle inequality follows from the triangle inequality of $L_p$ norms. The faithfulness and symmetry follow from the properties of the distance function $d(x,y)$. 

When $(\cX,d)$ is a Riemannian manifold equipped with the Riemannian distance $d$, the Wasserstein $2$-distance enjoys other equivalent formulations. One equivalent formulation is in terms of sub-solutions of the Hamilton Jacobi equation:
\begin{align}
    W_2^2(\mu,\nu)=\frac{1}{2}\inf\Big\{ \int u_1d\mu-\int u_0d\nu: \dot{u}_t+\frac{1}{2}|\nabla u_t|^2\le 0\Big\}
\end{align}
where $\nabla$ is the gradient operator. The second one is the Benamou-Brenier formula
\[W_2(\mu_0,\mu_1)=\inf\,\int_{0}^1\norm{\dot{\mu_t}}{g,\mu_t} \,dt\pl. \]
where the infimum is over all absolutely continuous paths $t\mapsto \mu_t$ connecting $\mu_0$ and $\mu_1$ in $\cD(\cX)$. The metric $\norm{\dot{\mu}}{g,\mu}$ is given by 
\begin{align}\label{eq:continuity} \norm{\dot{\mu}}{g,\mu}=\inf \Big\{ \int|v|^2d\mu \pl |\pl  \dot{\mu}+\text{div} (v\mu)=0 \Big\}\end{align}
where $\text{div}$ is the divergence operator. This formally expresses $W_2$ as the Riemannian distance of the metric \eqref{eq:continuity} on $\cD(\cX)$.

 \subsection{Non-commutative tansportation cost}
A large family of non-commutative transportation costs $\cT_{\cB}$ arises from duality to a semi-norm. Many are considered as quantum analogs of the Wasserstein $1$-distance in the literature, and are related to the concept of compact quantum metric spaces introduced by Rieffel \cite{rieffel2004compact}.
Let
 $\cA_\mathbb{R}\subset \cM_\mathbb{R}$ be a subspace and 
$L: \cA_\mathbb{R}\to [0,\infty)$ be a semi-norm.  Denote $\text{dom}(L):=\cA_{\mathbb{R}}$.
We consider the degenerate binary relation given by the unit ball of $L$: \[\cB_{L}=\{(x,x)\in \cA_\mathbb{R}\times \cA_\mathbb{R} | L(x)\le 1  \}\,.\]
The corresponding transportation cost is  
\[\cT_{L}(\omega_1,\omega_2)=\sup_{L(x)\,\le\, 1} (\omega_1-\omega_2)(x)\pl.\]
The following proposition follows by standard duality:
\begin{proposition}
$\cT_{L}:\cD(\cM)\times \cD(\cM)\to [0,\infty]$ is a pseudo-distance, i.e. it satisfies the triangle inequality and $\cT_{L}(\omega,\omega)=0$ for any $\omega$. If $\operatorname{dom}(L)$ is $w^*$-dense in $\cM_\mathbb{R}$, then $\cT_{L}(\omega_1,\omega_2)>0$ whenever $\omega_1\neq\omega_2$.
\end{proposition}
Often, the semi-norm above is defined as the restriction of a semi-norm $L:\cA\to [0,\infty)$ on a self-adjoint subspace $\cA=\cA^*\subset \cM$ satisfying $L(x)=L(x^*)$. In this case, we equivalently have 
\[\cT_{L}(\omega_1,\omega_2)=\sup_{L(x)\,\le\, 1} |\omega_1(x)-\omega_2(x)|\pl.\]
We start with three concrete examples.
\begin{example}(Trace distance)
The simplest instance of a non-commutative transportation cost distance arises by taking the operator norm and the degenerate binary relation $\cB:=\{(x,x)|\,x\in \mathcal{M}_{\mathbb{R}},\,\|x\|\le 1\}$. This recovers the usual dual norm on states:
\begin{align}
    \|\omega_1-\omega_2\|_1:=\sup_{\substack{ x=x^*\pl, \pl \|x\|\le 1}}\,\omega_1(x)-\omega_2(x)\,.
\end{align}
\end{example}
\begin{example}(Quantum Ornstein distance)\label{ornstein}
Recently, the following quantum generalization of the Ornstein distance was introduced in \cite{DePalma2021}: given a finite dimensional Hilbert space $\cH$ of dimension $d$ and an integer $n\in\mathbb{N}$, we consider the algebra $\cM:=\cB(\cH^{\otimes n})$ and the set $\cB_{\otimes }:=\{(x,x):\,x\in \cM_{\mathbb{R}},\,\|x\|_{\otimes}\le 1\}$, where the Ornstein Lipschitz constant $\|x\|_{\otimes}$ is defined as
\begin{align}
    \|x\|_{\otimes }:=\max_{i\in[n]}\min_{x^{(i)}}\,\|x-x^{(i)}\otimes 1_i\|\,,
\end{align}
and where the above minimization is over self-adjoint operators $x^{(i)}$ on the complement subsystem $\cH_{i^c}$ of $\cH_i$. The dual transportation cost distance is given by
\[\cT_{\otimes}(\rho, \si)=\min\Big\{\sum_{i}c_i\pl |\pl c_i\ge 0, \rho-\si=\sum_{i=1}c_i(\rho^{(i)}-\si^{(i)})\pl, \Tr_i(\rho^{(i)})=\Tr_i(\si^{(i)})\pl \forall i=1,\dots, n \Big\}\,.
\]
Here $\Tr_i$ denotes the partial trace map over $i$-th site.
\end{example}

\begin{example}(Oscillator norm)\label{ex:oscillator}
Alternatively, in \cite{Majewski1995} the authors considered the following semi-norm on $\cB(\cH^{\otimes n})$, also known as the oscillator norm:
\begin{align}
    \vertiii{x}_{\operatorname{osc}}:=\sum_{i=1}^n\,\|x-d^{-1}1_i\otimes \tr_{i}(x)\|\,,
\end{align}
so that $ \cB_{\operatorname{osc}}:= \big\{ (x,x)|\,x\in \cB(\cH^{\otimes n})_{\mathbb{R}},\, \vertiii{x}_{\operatorname{osc}}\le 1\big\}$.
\end{example}
The next two examples arise from non-commutative geometric structures.
\begin{example}(Compact quantum metric spaces)
Let $\cM$ be a unital $C^*$-algebra and $L:\cA \to [0,\infty)$ be a semi-norm defined on a self-adjoint norm-dense subspace $\cA\subset \cM$. The pair $(\cM,\cL)$ is said to be a compact quantum metric space if $\text{Ker}(L)=\mathbb{C}1$ and
the transportation cost distance $\cT_{L}$ metrizes the weak$^*$-topology on $\cD(\cM)$. $L$ is viewed as an abstract Lipschitz semi-norm which vanishes on $\mathbb{C}1$. In particular, if $\cM=C(\cX,d)$ is the algebra of continuous functions on a Polish space and $L(.)=\norm{\cdot}{\operatorname{Lip}}$ denotes the Lipschitz constant, $\cT_L$ reduces to the Wasserstein $1$-distance of Section \ref{sec:classical}.
\end{example}
\begin{example}(Spectral triples)
The origin of compact quantum metric spaces goes back to Connes's work (e.g. \cite{Connes1989,connes1990geometrie}) in non-commutative geometry.  A spectral triple $(\cA,\cH,D)$ consists of a $*$-algebra $\cA\subset \cB(\cH)$ of bounded operators acting on a Hilbert space $\cH$ and a (possibly unbounded) self-adjoint operator $D$ on $\cH$ such that the commutator $[D,a]$ is bounded for all $a\in\cA$. This self-adjoint operator $D$ models the spin Dirac operator on spin manifolds. The Lipschitz semi-norm is given by
\[ L(a):=\norm{[D,a]}{}\pl \quad \forall \pl a\in \cA\pl.\]
The distance
\[ \cT_{D}(\rho,\si)=\sup_{\norm{\  [D,a]\ }{}\le 1}|\rho(a)-\si(a)|\]
is also called Connes' distance.
\end{example}

The next example arises from gradient forms of quantum Markov semigroups.
\begin{example}(Gradient form)\label{exam:gradient}
Let $\cP_t=e^{t\cL}:\cM\to \cM$ be a quantum Markov semigroup with generator $\cL$. Recall that its gradient form is defined as:
\begin{align}\label{gradientform}
  2\,  \Gamma(x,x):=\cL(x^*x)-\cL(x^*)x-x^*\cL(x) \, , \pl x\in \operatorname{dom}(\cL)\pl.
\end{align}
In \cite{Junge2014}, the following semi-norms were introduced
\[ \norm{x}{\Gamma}=\max \{ \norm{\Gamma(x,x)}{}^{1/2},\, \norm{\Gamma(x^*,x^*)}{}^{1/2}\}, \quad x\in \cM\,,\]
where $\norm{\cdot}{}$ is the operator norm on $\cM$. The dual Wasserstein distance is defined as
\begin{align}
    W_\Gamma(\omega_1,\omega_2):=\sup\big\{(\omega_1-\omega_2)(x)|\,x\in\cM_{\mathbb{R}},\,\|\Gamma(x,x)\|\le 1 \big\}\,.
\end{align}
The norm $\norm{\cdot}{\Gamma}$ was used in \cite{Junge2014} in the study of $L_p$-Poincar\'{e} inequalities, as well as in \cite{gao2020fisher} to formulate a dual form to the transportation cost inequality for $W_\Gamma$ (see \Cref{sec:TC}). 
\end{example}

It is natural to define Lipschitz semi-norms through a differential structure. Let $(\cH, l,r, J,\partial)$ be a first order differential structure on $\cM$ as defined in section \ref{sec:differential}. Given a semi-norm $\norm{\cdot}{L}$ on the Hilbert space $\cH$, we define the semi-norm 
\[ \norm{x}{\partial,L}:=\norm{\partial x}{L}\]
whenever $\norm{\partial x}{L}$ is finite.

\begin{example}(First order differential structure) \label{exam:first2}
Given the finite dimensional matrix algebra $\cM=\mathbb{M}_n(\mathbb{C})$,
Carlen and Maas \cite{carlen2017gradient,Carlen2019} introduced a differential structure associated to a quantum Markov semigroup $ \cP_t=e^{t\cL}:\cM\to \cM$ which is symmetric with respect to a general faithful state $\si$, i.e.
$ \tau(\si \cP_t(x)^*y)=\tau(\si x^*\cP_t(y))$ for all $x,y\in\cM$.
Under this assumption, the Lindbladian  $\cL$ takes the form
\[ \cL(x)=\sum_{j\in \mathcal{J}}e^{-\omega_j}\Big(v_j^*[x,v_j]+[v_j^*,x] v_j\Big)\]
where $\mathcal{J}$ is a finite index set , $\omega_j\in \mathbb{R}$ and $\{v_j\}_{j\in \mathcal{J}}=\{v_j^*\}_{j\in \mathcal{J}}\subset \cB(\cH)$ is a self-adjoint set. The derivation is defined as
\[\partial=(\partial_j)_{j\in \mathcal{J}}:\cM\to \bigoplus_{j\in \mathcal{J}}\cM\pl,\quad  \partial_j(x):=[v_j,x] \pl. \]
The left (resp. right) actions are defined as the left (resp. right) multiplications
\[l(x)(V_j) =(xV_j)\pl, \quad \pl r(x)(V_j)=(V_jx)\pl.\]
Given such a gradient structure, we can define the semi-norm of $x\in\cM$ as (see also \cite{rouze2019concentration} for a slightly different definition):
\begin{align}\label{triple2norm}
\vertiii{x}_{\partial,2}:=\left(\sum_{j\in\mathcal{J}}\|\partial_j x\|_{\infty}^2\right)^{\frac{1}{2}}\,.
\end{align}
where $\norm{\cdot }{\infty}$ is the operator norm. 
\end{example}

We now discuss two quantum transportation costs which fit into our general framework.
\begin{example}{(Quantum Wasserstein 2-distance)}\label{exam:W2} The derivation in Example \ref{exam:first2} was used by
 Carlen and Maas \cite{carlen2017gradient} to introduce a quantum Wasserstein distance such that the reversible quantum Markov semigroup $t\mapsto e^{t\cL}$ on $\cM$ is the gradient flow of the quantum relative entropy to the equilibrium state. Indeed, the distance between two states $\omega_1$ and $\omega_2$ is given by
 \begin{align}\label{eq:Riemmannian}W_{\partial,2}(\omega_1, \omega_2):=\inf_{\gamma} \int_{0}^1\norm{\dot{\gamma}(t)}{g,\gamma(t)} dt\,,\end{align}
 where $\gamma:[0,1]\to \cD(\cM)$ is a piece-wise smooth curve such that $\gamma(0)=\omega_1$ and $\gamma(1)=\omega_2$. For a state $\rho\in \cD(\cM)$, the
 metric integrated above is defined as
 \begin{align}\label{eq:metricCM} \norm{x}{g,\rho}^2=\inf_{\mathbf{V}=(V_j)} \lan {\bf V}, [\rho]_{\vec{\omega}}{\bf V}\ran_{\tr} \, ,
\end{align}
where $[\rho]_{\vec{\omega}}:\cM\to \cM $ is the multiplication operator
\[ [\rho]_{\vec{\omega}}((V_j))=\Big([\rho]_{\omega_j}(V_j)\Big)\pl,\quad 
\pl [\rho]_{\omega_j}(X)=\int_0^1 e^{\omega_j(s-1/2)}\, \rho^s X\rho^{1-s}ds \pl.\]
The above infimum is over all $\mathbf{V}=(V_j)\in \bigoplus_{j\in \mathcal{J}}\cM$ satisfying $x=\partial^*\big(  [\rho]_{\vec{\omega}}{\bf V}\big)$. This infimum is an analog of the Benamou-Brenier formula for the classical Wasserstein 2-distance \cite{Ambrosio2012}. More recently, Wirth \cite{wirth2021dual} provided a dual formulation of this distance (albeit in the case of a primitive quantum Markov semigroup) : for any two states $\omega_1,\omega_2\in\cD(\mathbb{M}_n(\mathbb{C}))$: 
\begin{align}\label{CarlenMaas}
    \frac{1}{2}\,\mathcal{W}_{\partial,2}(\omega_1,\omega_2)^2=\sup\big\{\omega_1(a)-\omega_2(b):\,(a,b)\in\cB_{\partial} \big\}
\end{align}
where the set $\cB_{\partial}$ consists of all pairs of self-adjoint matrices $(a,b)$ that are connected by an a.e. differentiable path $\{\gamma(t)\}_{t\in[0,1]}$ in $\mathbb{M}_n(\mathbb{C})_\mathbb{R}$ that is a subsolution of the non-commutative Hamilton-Jacobi-Bellmann equation. Namely, for a.e. $t\in [0,1]$, and all $\rho\in \cD(\mathbb{M}_n(\mathbb{C}))$:
\begin{align}\label{eq:HJB}
    \rho(\dot{\gamma}(t))+\frac{1}{2}\|\partial \gamma(t)\|_{g,\rho}^2\le 0\,,
\end{align}
where the norm $\|.\|_\rho$ is the metric defined in \Cref{eq:metricCM}. 
The set $\cB_\partial$ satisfies the Definition \ref{def:metric}.
Indeed, since $\rho\in \cD(\mathbb{M}_n(\mathbb{C}))$ is arbitrary,
\eqref{eq:HJB} implies that $\mathcal{B}_\partial$ fulfills the conditions (i). The condition (ii) is satisfied with the trivial path $\gamma(t)=0$. Moreover, \Cref{eq:Riemmannian} justifies that $\mathcal{W}_{\partial,2}$ is indeed a distance. 

One can observe that for each $j$, the multiplication operator $[\rho_{j}]$ is an operator mean in section \ref{sec:operator}. Indeed,  \cite{Carlen2019,wirth2021dual} studied the quantum transport metric by replacing $[\rho]_{\vec{\omega}} $
by a family of general operator means \[\vec{\Lambda}(\rho):= (\Lambda_j(\rho))\pl,\quad  \pl\Lambda_j(\rho)=\Lambda_j(l(\rho),r(\rho))\,. \]
Both the gradient flow structure and the equivalent formation \eqref{eq:HJB} are obtained for the corresponding Wasserstein distance.
\end{example}
\begin{example}(Wasserstein distances via coupling) As discussed in \Cref{sec:classical}, transportation distances in  the commutative setting are naturally introduced as a minimization of the cost of a transport plan. Such an approach was recently extended to the non-commutative setting by various authors \cite{golse2016mean,de2021quantum,friedland2021quantum,cole2021quantum}. Given two states $\omega_1$, $\omega_2$ on $\cM$, we denote by $\Pi(\omega_1,\omega_2)$ the set of all couplings of $\omega_1$ and $\omega_2$, i.e.~:
\begin{align}
    \pi(\omega_1,\omega_2):=\big\{\pi\in\cD(\cM\overline{\otimes} \cM):\,\pi(x\otimes 1)=\omega_1(x),\,\pi(1\otimes y)=\omega_2(y) \,\pl \forall x,y\in\cM\big\}\,.
\end{align}
Here and in the following $\cM\overline{\otimes} \cM$ denote the von Neumann algebra tensor product.
Next, given a positive element $C\in \cM\overline{\otimes} \cM$, we define the transportation cost as
\begin{align}
    \cT_C(\omega_1,\omega_2):=\inf_{\pi\in\Pi(\omega_1,\omega_2)}\,\pi(C)\,.
\end{align}For the finite dimensional matrix algebra $\cM=\mathbb{M}_n(\mathbb{C})$, the above transportation cost enjoys a dual formulation similar to the classical setting (see \cite[Theorem 2.2]{caglioti2021towards}): for any two states $\omega_1,\omega_2$,
\begin{align}
    \cT_C(\omega_1,\omega_2)=\sup\big\{\omega_1(x)-\omega_2(y):\,x\otimes 1-1\otimes y\le C,\,x,y\in\mathbb{M}_n(\mathbb{C})_{\mathbb{R}} \big\}\,.
\end{align}
Recall that a binary function $d:\mathcal{X}\otimes \mathcal{X}\to [0,\infty)$ is a weak metric if it satisfies
\begin{enumerate}
\item[i)] $d(x,y)=d(y,x)$
\item[ii)] $ d(x,y)\ge 0$, and $d(x,y)=0$ if and only if $x=y$.
\end{enumerate} 
It was proved in 
\cite[Theorem 5.2]{cole2021quantum} that 
$\cT_C$ is a weak metric on $\cD(\cM)$ if and only if it is strictly supported on the anti-symmetric subspace \[\cH_A=\text{span}\{\ket{\psi}\otimes \ket{\varphi}- \ket{\varphi}\otimes \ket{\psi} | \ket{\psi}, \ket{\varphi}\in \mathbb{C}^n\}\,.\]
To see why the binary relation $\cB_C:=\{(x,y)| x\otimes 1-1\otimes y\le C,\,x,y\in\mathbb{M}_n(\mathbb{C})_{\mathbb{R}}  \}$ satisfies the axioms of Definition \ref{def:metric}, we introduce
the swap unitary operation $S$ through the relation $S(\ket{\psi}\otimes \ket{\varphi})=\ket{\varphi}\otimes \ket{\psi}$. The anti-symmetric subspace $\cH_A$ is the eigenspace of $S$ corresponding to the eigenvalue $-1$. First, the non-negativity condition (ii) is trivially satisfied with $x=y=0$. To prove (i), we note that for any pure state $\phi=\ket{\phi}\bra{\phi}$, \[\phi\ten \phi(C)= (\bra{\phi}\ten\bra{\phi})C(\ket{\phi}\ten \ket{\phi})=0\pl\]
because $\ket{\phi}\ten \ket{\phi}$ is in the symmetric subspace and hence orthogonal to $\cH_A$. Then by $SCS=C$, we have that for any $x,y\in\mathbb{M}_n(\mathbb{C})_\mathbb{R}$ satisfying $x\otimes 1-1\otimes y\le C$,
\begin{align}
    (x-y)\otimes 1-1\otimes (x-y)\le 2C\,,
\end{align}
which implies that $x\le y$ after evaluating the inequality on tensor product pure states $\phi\otimes \phi$.
\end{example}
As opposed to the previously discussed examples, the transportation cost $\cT_C$ (or $\cT_C^{1/p}$ to some $p>1$) is not always a distance. It was proved in \cite[Corollary 8.3]{cole2021quantum} that $\sqrt{T_{C}}$ is a distance on qubit system (n=2). In general, the lack of triangle inequality is partially due to the failure of the gluing Lemma \ref{lemma:glue} in the quantum case. For example, there is not tripartite state $\rho\in \cD(\mathbb{M}_n(\mathbb{C})\ten \mathbb{M}_n(\mathbb{C})\ten \mathbb{M}_n(\mathbb{C}))$ such that both reduced densities $\rho_{12}$ and $\rho_{23}$ are maximally entangled state.

\section{Coarse Ricci curvature and gradient estimate}\label{sec:curvature}
In this section, we discuss a notion of non-commutative curvature of a quantum channel based on the general transportation cost introduced in \Cref{sec:NCTC}. We recall that a quantum channel $\mathcal{P}:\cM\to \cM$ is a normal completely positive unital map and its pre-adjoint $\mathcal{P}^\dagger$ is a transformation on the state space $\cD(\cM)\to \cD(\cM)$ given by $\cP^\dag(\omega)=\omega\circ \cP$. By slight abuse of notation, we will also refer to $\cP^\dagger$ as a quantum channel. Recall that the multiplicative domain of $\cP$ is defined as 
\begin{align}
    \cM(\cP):=\big\{a\in\cM,\,\cP(a^*a)=\cP(a^*)\cP(a),\,\cP(aa^*)=\cP(a)\cP(a^*)\big\}\,.
\end{align}
and we define its decoherence-free subalgebra as $\cN(\cP):=\cap_{n\ge 1} \cM(\cP^n)$. It is clear that the restriction of $\cP$ to $\cN(\cP)$ is a $*$-homomorphism. Here, we also assume that there exists a normal conditional expectation $E_\cN:\cM\to \cN(\cP)$ onto $\cN(\cP)$ \cite{Carbone2019}.

\begin{definition}\label{def:corasericcicurvature}
Let $(\cM,\cB)$ be a non-commutative transportation metric and $\cP:\cM\to\cM$ a quantum channel. The non-commutative coarse Ricci curvature of the triple $(\cM,\cB,\cP)$ at a pair of states $(\omega_1,\omega_2)\in\operatorname{dom}(\cT_\cB)$ with $\omega_1\circ E_\cN=\omega_2\circ E_\cN$ is defined as
\begin{align}
    \kappa(\omega_1,\omega_2):=1-\frac{\cT_\cB(\omega_1\circ \cP,\omega_2\circ \cP)}{\cT_\cB(\omega_1,\omega_2)}\,.
\end{align}
with the convention that $\kappa(\omega_1,\omega_2)=-\infty$ whenever $\cT_\cB(\omega_1\circ \cP,\omega_2\circ \cP)=\infty$. We say that the triple $(\cM,\cB,\cP)$ has the non-commutative coarse Ricci curvature lower bound $\kappa\in\mathbb{R}$ if for any two states $\omega_1,\omega_2\in \cD(\cM)$ with $\omega_1\circ E_\cN=\omega_2\circ E_\cN$, 
$\kappa(\omega_1,\omega_2)\ge \kappa$, i.e.
\begin{align}\label{eq:curvature}
    \cT_\cB(\omega_1\circ \cP,\omega_2\circ \cP)\le (1-\kappa) \, \cT_\cB(\omega_1,\omega_2)\,.
\end{align}
\end{definition}
The above definition is motivated from Ollivier's Ricci curvature \cite{Ollivier2009} of classical Markov chains with respect to the Wasserstein 1-distance discussed in Section \ref{sec:classical}, as well as from 
the quantum entropic curvature lower bound of a quantum Markov semigroup as introduced in \cite{carlen2014analog,carlen2017gradient}, where $\cT_\cB$ is taken as the quantum Wasserstein $2$-distance $\mathcal{W}_{\partial,2}$ defined in Example \ref{exam:W2}

\begin{rem}As we will be mostly interested in the case of positive coarse Ricci curvature, the condition that $\omega_1\circ E_\cN=\omega_2\circ E_\cN$ allows us to consider maps $\cP$ with multiple invariant states. Indeed, in the generic setting $\omega\circ \cP^n-\omega\circ E_\cN\to 0$, 
if the curvature bound $\kappa>0$, we have for any $\cT_\cB(\omega_1,\omega_2)<\infty$,
\begin{align*}
    \cT_\cB(\omega_1\circ E_{\cN},\omega_2\circ E_{\cN})= 
   \lim_{n}\cT_\cB(\omega_1\circ \cP^n,\omega_2\circ \cP^n)\le \lim_{n}(1-\kappa)^n\,\cT_\cB(\omega_1,\omega_2)=0
\end{align*}
which implies $\omega_1\circ E_{\cN}=\omega_2\circ E_{\cN}$ if $\cT_\cB$ is faithful.


\end{rem}

In the spirit of \cite{Ollivier2009}, we give two simple properties of the coarse Ricci curvature introduced in \Cref{def:corasericcicurvature}.
\begin{proposition}[(Composition and superposition)]
    Let $\cT_\cB$ be a transportation cost on the von Neumann algebra $\cM$. Let $\cP_1,\cP_2:\cM\to \cM$ be two quantum channels such that the triple $(\cM,\cB,\cP_1)$, resp. $(\cM,\cB,\cP_2)$, has coarse Ricci curvature lower bound $\kappa_1\in\mathbb{R}$, resp.~$\kappa_2\in\mathbb{R}$.  Then 
\begin{enumerate} 
\item[$\operatorname{(i)}$] (Composition): $\cP_1\circ\cP_2$ has coarse Ricci curvature lower bound $\kappa_1+\kappa_2-\kappa_1\kappa_2$.  
\item[$\operatorname{(ii)}$] (Superposition): for any $0\le \la\le 1$, $\la\cP_1+(1-\la)\cP_2$ has coarse Ricci curvature lower bound $\la\kappa_1+ (1-\la)\kappa_2$.
\end{enumerate}
\end{proposition}
\begin{proof}
(i) follows from the definition and (ii) is clear from the joint convexity of $\cT_\cB$.
\end{proof}

\subsection{From gradient estimate to curvature}
We discuss two gradient type estimates that imply coarse Ricci curvature bounds. The first one is for the transportation cost $\cT_L$ arising from a semi-norm $L$. We assume that $\text{ker}(L)$ is a subalgebra $\cN$. Then $W_L(\omega_1,\omega_2)=\infty$ whenever $\omega_1\circ E_\cN\ne \omega_2\circ E_\cN$. This is because there exists $x\in \cN$ with $L(x)=0$ such that $\omega_1(x)=\omega_1\circ E_\cN(x)\neq \omega_2\circ E_\cN(x)=\omega_2(x)$. 

We say a quantum channel $\cP$ satisfies the $\kappa$-Lipschitz estimate for $\kappa\in \mathbb{R}$ if  for any $x\in \dom(L)$, $\cP(x)\in  \dom(L)$ and \[L(\cP(x))\le (1-\kappa) L(x)\pl.\]
It is clear that if $\cP$ satisfies $\kappa$-Lipschitz estimate for any $\kappa$, then $\cP(\cN)\subset \cN$. In the classical setting, the Lipschitz estimate and coarse Ricci curvature (w.r.t to $W_1$ distance) are equivalent. Here in the abstract non-commutative setting, one direction is clear. 

\begin{proposition}\label{curvaturetolip}
If $\cP$ satisfies $\kappa$-Lipschitz estimate for the semi-norm $L$, then the triple $(\cM,\cB_L,\cP)$ has coarse Ricci curvature lower bound. 
\end{proposition}
\begin{proof}
Direct by duality. 
\end{proof}

One important property of the classical coarse Ricci curvature is the $L_1$-tensorization (see \cite[Proposition 27]{Ollivier2009}). Here we discuss $L_1$-tensorization of Lipschitz estimates. As for other functional inequalities \cite{bardet2018hypercontractivity,gao2021spectral}, the tensorization property does not naturally extend to the non-commutative framework. In order to recover it, we use the notion of a matrix Lipschitz semi-norm (see e.g. \cite{wu2004non}). We let $\cA \subset \cM$ be a $w^*$-dense subalgebra and we denote $\mathbb{M}_n(\cA)=\cA\otimes \mathbb{M}_n$ for each integer $n\ge 1$ and similarly for $\cM$. A matrix Lipschitz semi-norm is a family of semi-norms $L^{(n)}:\mathbb{M}_n(\cA)\to [0,\infty)$ such that for all $n,m\ge 1$
\begin{enumerate}
\item[(i)] $L^{(n)}(1)=0$ ;

\item[(ii)] for all $x\in \mathbb{M}_n(\cA)$, $L^{(n)}(x)=L^{(n)}(x^*)$;

\item[(iii)] for all $x\in \mathbb{M}_n(\cA),y\in \mathbb{M}_m(\cA)$, $${L^{(n+m)}}\left(\left[\begin{array}{cc}
     x  & 0\\
     0&  y
\end{array}\right]\right)=\max \{{L^{(n)}}({x}), {L^{(m)}}({y}) \}\,;$$
\item[(iv)] for all $x\in \mathbb{M}_n(\cA),\,a\in \mathbb{M}_{m,n}(\mathbb{C})$ and $b\in \mathbb{M}_{n,m}(\mathbb{C})$, ${L^{(m)}}({axb})\le \|a\|\, {L^{(n)}({x}) }  \|{b}\|$.
\end{enumerate}
Next, we introduce a notion of complete boundedness for Lipschitz semi-norms: given a matrix Lipschitz semi-norm $(\mathbb{M}_n(\cA),L^{(n)})_{n\ge 1}$, 
we say a quantum channel $\cP:\cM\to \cM$ satisfies the $\kappa$-\textit{complete Lipschitz estimate} for $\kappa\in\mathbb{R}$ 
if for any $n\ge 1$ and any $x\in \mathbb{M}_n(\cA)$,
\begin{align}
   {L^{(n)}}( (\id\otimes \cP)(x))\le (1-\kappa)\,{L^{(n)}}( x)\,.
\end{align}
With the above definition, the $L_1$ tensorization extends naturally:
\begin{proposition}
Let $(\cM_i,\cB_{L_i})_{i\in \mathcal{I}}$ be a finite family of non-commutative metric spaces induced by the matrix Lipschitz semi-norms $(L_i)_{i\in \mathcal{I}}$. Assume that for each $i\in \mathcal{I}$, $\cP_i:\cM_i\to \cM_i$ is a quantum channel satisfying the $\kappa_i$-complete Lipschitz estimate of parameter $\kappa_i\ge 0$.  Denote by $\widetilde{L}$ the Lipschitz semi-norm on $\bigotimes_{i\in \mathcal{I}}\cA_i$ as
\begin{align}
    \widetilde{L}({x})=\sum_{i\in\mathcal{I}}\,L_i({x})\,. 
\end{align} Then for any probability distribution $\{\alpha_i\}_{i\in\mathcal{I}}$,  the quantum channel
\begin{align}
    \widetilde{\cP}:=\sum_{ i\in\mathcal{I}}\alpha_i\,\cP_i\otimes \id_{i^c}
\end{align}
satisfies the $\kappa$-complete Lipschitz estimate of parameter $\kappa=\min_i\alpha_i\kappa_i$ with respect to the semi-norm $\widetilde{L}$.  Here $\id_{i^c}$ is the identity map on $\bigotimes_{j\neq i}\cM_j$.
\end{proposition}
\begin{proof}
For the ease of notation, we simply write $L^{(n)}=L$ for a semi-norm $L$.
We have 
\begin{align}
  {\tilde{L}}( \widetilde{\cP}(x))&\le \sum_{i}\alpha_i\sum_{j} {L_{i}}((\cP_j\ten \id) (x))\nonumber\\
   &\le \sum_{i}\alpha_i\,{L_{i}}((\cP_i\ten \id) (x))+\sum_{i\neq j}\alpha_i\,{L_{j}}((\cP_i\ten \id) (x))\nonumber
   \\  &\le \sum_{i}\alpha_i(1-\kappa_i)\,{L_{i}}(x)+\sum_{i\neq j}\alpha_i\,{L_{j}}(x)\label{eq:contraction}
    \\  &\le \sum_{i,j}\alpha_i\,{L_{j}}(x)-\sum_{i}\alpha_i\kappa_i\,{L_{i}}(x)\nonumber
\\  &\le (1-\min_{i}\alpha_i\kappa_i){\tilde{L}}({x})\,.\nonumber\end{align}
In Equation \eqref{eq:contraction}, we used the fact that ${L_{j}}((\cP_i\ten \id )(x))\le {L_{j}}({x})$. This follows from (iv) in the definition of a matrix Lipschitz semi-norm for ${L_j}$ and the fact that $\cP_i$ is a complete contraction on $\cM_i$.
\end{proof}


We shall now consider the analog of gradient estimate in \cite{wirth2020complete} which implies the contraction of quantum transport metrics in Example \ref{exam:W2}. Let $\cM$ be a finite von Neumann algebra and $(\cA,\cH, l,r, \partial)$ be a first order differential structure on $\cM$ as defined in section \ref{sec:differential}. Given an operator mean $\Lambda$, we consider the multiplication operator for a positive operator $\rho$,
\begin{align}\label{eq:mean1}\Lambda_\rho: \cH\to \cH\pl,\qquad\qquad \Lambda_\rho=\Lambda(l(\rho),r(\rho)) \end{align}
where $l$ (resp. $r$) is the left action (resp. right action). When $\cH\cong \bigoplus_{j\in \mathcal{J}}L_2(\cM)$ is a direct sum of copies of $L_2(\cM)$, we can further consider 
\begin{align}\label{eq:mean2}\Lambda_\rho=(\Lambda_{\rho,j}):\bigoplus_{j\in \mathcal{J}}L_2(\cM)\pl, \pl\quad  \Lambda_{\rho,j}=\Lambda_{j}(l(\rho),r(\rho))\pl,  \end{align}
for a family of operator mean $\Lambda_j,j\in \mathcal{J}$. The simpler setting \eqref{eq:mean1} is sufficient for symmetric quantum Markov semigroup on finite von Neumann algebra \cite{wirth2018noncommutative} and the second setting \eqref{eq:mean2} is needed for finite dimensional GNS-symmetric semigroups as in \cite{carlen2017gradient,Carlen2019,wirth2021dual} .
In both cases, we define the following weighted norm on $\cH$:
\begin{align}\label{eq:weight} \norm{x}{\rho}:=\langle V, \Lambda_\rho (V)\rangle_\cH\,,\qquad V\in\cH\,.\end{align}
Recall that the associated quantum Wasserstein 2-distance is given by
 \begin{align}\label{eq:Riemmannian}W_{\partial,2}(\omega_1, \omega_2):=\inf_{\gamma} \int_{0}^1\norm{\dot{\gamma}(t)}{g,\gamma(t)} \,dt\,,\end{align}
 where $\gamma:[0,1]\to \cD(\cM)$ is an absolutely continuous curve such that $\gamma(0)=\omega_1$ and $\gamma(1)=\omega_2$. The above integrated metric is defined as
 \begin{align}\label{metricCM} \norm{x}{g,\rho}^2:=\inf \lan V, \Lambda_\rho (V)\ran_{\cH} \, ,
\end{align}
where the infimum is over all $V\in \cH$ satisfying $x=\partial^*\big( \Lambda_\rho (V)\big)$. Let $t\mapsto \cP_t=e^{t\cL}$ be the quantum Markov semigroup generated by the derivation $\cL=-\partial^*\partial$. In finite dimensions, it is proved that the following are equivalent
\begin{enumerate}
\item[(i)] for each $t$, $\cP_t$ satisfies the gradient estimate $\norm{\partial\cP_t(x)}{\rho}\le e^{-\lambda t}\norm{\partial x}{\cP_t^{\dag}(\rho)}$ for all $x\in\cA$ and all $\rho\in\cD(\cM)$;
\item[(ii)] for any $\omega_1, \omega_2\in \cD(\cM)$, $W_{\partial,2}(\omega_1\circ \cP , \omega_2\circ \cP)\le e^{-2\la t} \,W_{\partial,2}(\omega_1, \omega_2)$.
\end{enumerate}
We remark that the implication $(i)\Rightarrow (ii)$ is also obtained for primitive symmetric quantum Markov semigroups on finite von Neumann algebras.

Motivated from the above, we introduce the gradient estimate for quantum channels. We say a quantum channel $\cP:\cM\to\cM$ satisfies the gradient estimate of constant $\kappa\in\mathbb{R}$ (in short, $\kappa$-GE) if for any $x\in \cA$ and density operator $\rho$, 
\[\norm{\partial\cP(x)}{\rho}\le (1-\kappa)\norm{\partial x}{\cP^{\dag}(\rho)} \]
where $\cP^\dagger(\rho)=\rho\circ \cP$ is the adjoint map of $\cP$ on $\cD(\cM)$. 

\begin{proposition}\label{prop:equivalence}
$\cP$ satisfies $\kappa$-$\operatorname{GE}$ implies that $\cP$ satisfies $\kappa$-coarse curvature bound for $W_{\partial,2}$, i.e.
for any $\omega_1, \omega_2\in \cD(\cM)$,
\begin{align*}
W_{\partial,2}(\omega_1\circ \cP , \omega_2\circ \cP)\le (1-\kappa ) W_{\partial,2}(\omega_1, \omega_2)\,.
\end{align*}
\end{proposition}
For the proof, We need the following lemma.
\begin{lemma}\label{lemma:operator}Let $A,B$ be densely defined positive operators on some Hilbert space $\cH$ and $C\in\cB(\cH)$.
Suppose for some $\lambda>0$,  $C^*AC\le  \lambda B$. Then  $CB^{-1}C^*\le \lambda A^{-1}$, where $A^{-1},B^{-1}$ are the inverse operator on the corresponding supports.
\end{lemma}
\begin{proof}We write $\dom(B)$ as the domain of $B$ and $\text{supp}(B)$ as the support of $B$.
$C^*AC\le \lambda B$ means that for any $\xi\in \dom(B)$, $\lan \xi, C^*AC\xi \ran_\cH\le \lambda  \lan \xi, B\xi \ran_\cH$. Then for any $\xi\in \dom(B)\cap \text{supp}(B)$,
\[ \lan \xi, B^{-1/2}C^*AC B^{-1/2}\xi \ran_\cH\le \lambda \norm{\xi}{\cH}^2,\]
which implies $\norm{B^{-1/2}C^*AC B^{-1/2}}{}\le \lambda$ as a bounded operator on $\cH$. This further implies
$\norm{B^{-1/2}C^*A^{1/2}}{}\le \sqrt{\lambda}$, and hence for $\xi \in \cH$,
\[ \lan \xi, A^{1/2}CB^{-1}C^* A^{1/2}\xi \ran_\cH\le \lambda \norm{\xi}{\cH}^2\]
Then for any $\eta\in \text{ran}(A^{1/2})=\dom(A)\cap \supp(A)$, we have
\[ \lan \eta , CB^{-1}C^*\eta \ran_{\cH}\le \la \lan \eta, A^{-1}\eta \ran_\cH\pl.\]
\end{proof}

\begin{proof}[Proof of Proposition \ref{prop:equivalence}]
Given an invertible density operator $\rho\in\cD(\cM)$, denote the positive densely defined operator $\partial^*\Lambda_{\rho}\partial=M_\rho$ on $\cA\subset L_2(\cM)$. Note that $\cP^\dag$ is also the adjoint map of $\cP$ on $L_2(\cM)$ with respect to the trace inner product. By Lemma \ref{lemma:operator}, 
we have 
\begin{align*}
\text{$\kappa$-GE}\,\Longleftrightarrow\,  & \, \cP^\dagger M_\rho\cP \le (1-\kappa)^2  M_{\cP^\dagger(\rho)}
\\ \,\Longrightarrow \, & \, \cP M_{\cP^\dagger(\rho)}^{-1}\cP^\dagger\le (1-\kappa)^2  M_{\rho}^{-1}\pl. 
\end{align*}
Note that from  \cite[Definition 7.6]{Carlen2019}, 
\begin{align}\label{metricCM} \norm{x}{g,\rho}^2:=&\inf_{x=\partial^*(\Lambda_\rho( V))} \lan V, \Lambda_\rho (V)\ran_{\cH} =\lan x , M_{\rho}^{-1} x\ran_{L_2(\cM)}\,.
\end{align}
Given an absolutely continuous path $\gamma:[0,1]\to \cD(\cM)$ connecting $\omega_1$ and $\omega_2$, $\cP^\dagger(\gamma)$ is a
path connecting $\cP^\dagger(\omega_1)$ and $\cP^\dagger(\omega_2)$. Then for each $t\in [0,1]$ ,
\begin{align*} \norm{\cP^\dagger(\gamma'(t))}{g,\cP^\dagger(\gamma(t))}^2&= \lan \cP^\dagger(\gamma'(t)) , M_{\cP^\dagger(\gamma(t))}^{-1} \cP^\dagger(\gamma'(t))\ran
\\ &\le  (1-\kappa)^2 \lan \gamma'(t) , M_{\gamma(t)}^{-1} \gamma'(t)\ran\\
&= (1-\kappa)^2 \norm{\gamma'(t)}{g,\gamma(t)}^2
\end{align*}
Integrating over $t\in [0,1]$ for an arbitrary such curve $\gamma$ yields the assertion.



\end{proof}

For the tensorization property, we introduce the complete notion of gradient estimate. 
We say that $\cP$ satisfies the complete gradient estimate of constant $\kappa\in\mathbb{R}$ ($\kappa$-CGE in short) if $\id_{\mathcal{R}}\ten \cP$ satisfies $\kappa$-GE for the derivation $\id \ten \partial: \mathcal{R}\otimes\cA \to L_2(\mathcal{R})\otimes\cH$ and for all finite von Neumann algebra $\mathcal{R}$. 

\begin{proposition}
Let $(\cM_i,\tau_i), i=1,2$ be finite von Neumann algebras equipped with corresponding first order structures $(\cA_i,\cH_i, l_i,r_i, \partial_i)$. Let $\cP_i:\cM_i\to \cM_i$ be quantum channels. If for $i=1,2$, $\cP_i$ satisfies $\kappa$-$\operatorname{CGE}$
for $\partial_i$ respectively, then $\cP_1\ten \cP_2$ satisfies $\kappa$-$\operatorname{CGE}$ for the derivation
\[ \partial=\partial_1\ten \id \oplus \id \ten \partial_2: \cA_i\otimes \cA_2 \to (\cH_1\ten L_2(\cM_2))\oplus (L_2(\cM_1)\ten \cH_2)\pl.\]
\end{proposition}
\begin{proof}
The proof is similar to \cite[Theorem 4.1]{wirth2020complete}. The details are left to the reader.
\end{proof}

\subsection{From curvature to spectral gap}
In this section, we prove that the Lipchitz estimate of a quantum channel $\cP$ lower bounds its spectral gap. Let $\cP:\cM\to \cM$ be a quantum channel that is GNS-symmetric with respect to a faithful state $\omega\in\cD(\cM)$, i.e. for any $x,y\in \cM$:
\[ \omega(\cP(x)y)=\omega(x\cP(y))\,.\]
It follows that $\omega$ is invariant under $\cP$, i.e. $\omega\circ \cP=\omega$. 
By Kadison-Schwarz inequality, $\cP$ is a contraction on the $\omega$-weighted $L_2$ norm $\|x\|_{L_2(\omega)}:=\omega(x^*x)$ and $\cP$ is an isometry on the multiplicative domain of $\cP$, which we denote by $\cN$. 
Recall that the spectral gap of $\cP$ is defined as
\[\la(\cP):=\inf_{E_{\cN}(x)=0} \frac{\omega(\cP(x)^*\cP(x))}{\omega(x^*x)} \pl. \]
We need the following non-commutative $(2-\infty)$ Poincar\'{e} inequality:
\begin{definition}
We say a semi-norm $L$ with $w^*$-dense domain $\cA$ satisfies a $(2-\infty)$ Poincar\'{e} inequality with respect to an invariant state $\omega$ and subalgebra $\cN$ if there exists a constant $C\ge 0$ such that for any $x\in\cA$,
\begin{align}
     \norm{x-E_\cN(x)}{L_2(\omega)}\le C \,L(x)\,.
\end{align}
\end{definition}
Similar estimates were considered in the setting of non-commutative diffusions under a Bakry-Emery condition in \cite{Junge2014}. In what follows, we prove that under the assumption of a $(2-\infty)$ Poincar\'{e} inequality, the non-commutative coarse Ricci curvature provides a lower bound on the gap of $\cP$. 
\begin{proposition}
Let $\cP$ be a quantum channel with multiplicative domain $\cN$ and symmetric with respect to a faithful state $\omega\in\cD(\cM)$. Assume that the semi-norm $L$ satisfies a $(2-\infty)$-Poincar\'{e} inequality with respect to $\omega$ and subalgebra $\cN$. If $\cP$ satisfies the gradient estimate $L({\cP(x)})\le (1-\kappa)L({x})$ with respect to the semi-norm $L$ with $\kappa>0$, then $\lambda(\cP)\ge \kappa$. 
\end{proposition}
\begin{proof}
For $x\in\cA$ and any $n\in\mathbb{N}$, we have by the $(2-\infty)$ Poincar\'{e} inequality that  
\begin{align}
\norm{\cP^n(x)-E_\cN(\cP^n(x))}{L_2(\omega)}\,\le\, C L(\cP^n(x))\,\le\, C\,(1-\kappa)^{n}\,L({x})\,.
\end{align}
Note that $\cP\circ E_\cN=E_{\cN}\circ\cP$.
Taking the $n$-th root of the above inequality, we have 
\begin{align}
    \lim_{n\to\infty}\|\cP^n(x-E_\cN(x))\|_{L_2(\omega)}^{\frac{1}{n}}\le 1-\kappa\,.
\end{align}
In other words, the spectral radius of $\cP$ on $L_2(\omega)$ is at most $1-\kappa$ on the domain $\cA$ of $L$. 
We conclude the proof using that $\cP$ is a bounded self-adjoint operator on $L_2(\omega)$, so that it is sufficient to consider its spectral radius on any dense subspace.
\end{proof}
\subsection{From curvature to transportation cost inequalities}\label{sec:TC}
Let $E_i:\cM\to \cN_i,i=1,\dots, n$ be a finite family of conditional expectations. We consider the quantum channel $\cP=\frac{1}{n}\sum_{i=1}^nE_i$ as the average of the channels $E_i$. The invariant subalgebra of $\cP$ is $\cN=\cap_i\cN_i$. The following result is inspired by \cite{Eldan2017} (see also \cite{de2021quantumb} for a slightly different version of the result):

\begin{theorem}
Assume that the triple $(\cM,\cB,\cP=\frac{1}{n}\sum_{i=1}^nE_i)$ has coarse Ricci curvature bound $\kappa>0$. Suppose further that for each $i\in[n]$, the transportation cost $\cT_\cB$ satisfies triangle inequality and
\begin{align}
    \cT_\cB(\rho,\rho\circ E_i)\le C\,\|\rho-\rho\circ E_i\|_1\,.
\end{align}
Then for any $\rho\in \cD(\cM)$,
\begin{align}
    \cT_\cB(\rho,\rho\circ E_\cN)\le \frac{C}{1-(1-\kappa)^n}\,\sqrt{2n\,D(\rho\|\rho\circ E_\cN)}\,.
\end{align}
\end{theorem}
\begin{proof}
For any $j\in\mathbb{N}$, denote $\rho^{(j)}:=\rho\circ\cP^{j}$. Then, by the triangle inequality and for $N\in\mathbb{N}$, we have 
\begin{align*}
    \cT_\cB(\rho,\rho\circ E_\cN)&\le \sum_{j=1}^N\,\cT_\cB(\rho^{(j-1)},\rho^{(j)})+\cT_\cB(\rho^{(N)},\rho \circ E_\cN)\\
  &\le \frac{1}{n}\,\sum_{i=1}^n\sum_{j=1}^N\,\cT_\cB(\rho^{(j-1)},\rho^{(j-1)}\circ E_i) +(1-\kappa)^N\cT_\cB(\rho,\rho\circ E_\cN)\\
  &\le \frac{C}{n}\sum_{i=1}^n\sum_{j=1}^N\,\|\rho^{(j-1)}-\rho^{(j-1)}\circ E_i\|_1+(1-\kappa)^N\cT_\cB(\rho,\rho\circ E_\cN)\\
  &\le \frac{C\sqrt{2}}{n}\sum_{i=1}^n\sum_{j=1}^N\,\sqrt{D(\rho^{(j-1)}\|\rho^{(j-1)}\circ E_i)}+(1-\kappa)^N\cT_\cB(\rho,\rho\circ E_\cN)\,,
\end{align*}
where the last line follows from Pinsker's inequality. Choosing $N=n$ and rearranging the terms, we have by Jensen's inequality that
\begin{align*}
    \big(1-(1-\kappa)^n\big)\, &\cT_\cB(\rho,\rho\circ E_\cN)\\ &\le  {C\sqrt{2}}\Big(\sum_{i,j=1}^nD(\rho^{(j-1)}\|\rho^{(j-1)}\circ E_i)\Big)^{\frac{1}{2}}\\
    &\overset{(1)}{=} C\sqrt{2}\,\Big(\sum_{i,j=1}^nD(\rho^{(j-1)}\|\rho\circ E_\cN)-D(\rho^{(j-1)}\circ E_i\|\rho\circ E_\cN) \Big)^{\frac{1}{2}}
    \\
    &= C\sqrt{2n}\,\Big(\sum_{j=1}^n \big(D(\rho^{(j-1)}\|\rho\circ E_\cN)-\frac{1}{n}\sum_{i}D(\rho^{(j-1)}\circ E_i\|\rho\circ E_\cN)\big) \Big)^{\frac{1}{2}}
    \\
    &\overset{(2)}{\le }\,C\sqrt{2n}\,\Big(\sum_{j=1 }^n D(\rho^{(j-1)}\|\rho\circ E_\cN)-D(\rho^{(j)}\|\rho\circ E_\cN)\Big)^{\frac{1}{2}}
   \\
    &\le C\sqrt{2n\,D(\rho\|\rho\circ E_\cN)}\,,
\end{align*}
where $(1)$ follows from the chain rule for the relative entropy \eqref{eq:chain} whereas $(2)$ follows from the joint convexity of the relative entropy.
\end{proof}




\subsection{From curvature to transportation information inequalities}
In this subsection, we consider the semi-norm discussed in Example \ref{exam:first2}. Let $ t\mapsto \cP_t=e^{t\cL}: \mathbb{M}_n(\mathbb{C})\to \mathbb{M}_n(\mathbb{C})$ be a quantum Markov semigroup that is GNS-symmetric with respect to a faithful state $\si$. We recall the definition of the derivation and semi-norm
 \begin{align*}\partial=(\partial_j)_{j\in \mathcal{J}}:\mathbb{M}_n(\mathbb{C})\to \bigoplus_{j\in \mathcal{J}}\mathbb{M}_n(\mathbb{C})\pl,\quad  \partial_j(x)=[v_j,x] \pl,\quad\vertiii{x}_{\partial,2}=\Big(\sum_{j\in J}\|{\partial_j x}\|^2\Big)^{1/2} \,,\end{align*}
 where $\norm{\cdot}{}$ is the operator norm and $v_j$ are the Lindbald operators of the Lindbladian
 \begin{align}\label{generatorGNSsymm}
      \cL(x)=\sum_{j\in \mathcal{J}}e^{-\omega_j}\Big(v_j^*[x,v_j]+[v_j^*,x] v_j]\Big)\,.
 \end{align}
 Define the multiplication operator $\Gamma_{\sigma}(x)=\sigma^{1/2}x\sigma^{1/2}$ as a positive operator on $L_2(\mathbb{M}_n(\mathbb{C}))$. The $\sigma$-weighted $L_p$ norm is $\|x\|_{L_p(\sigma)}:=\norm{ \Gamma_\sigma^{1/p}(x)}{p}= \tr(|\si^{\frac{1}{2p}} x\si^{\frac{1}{2p}}|^{p})^{\frac{1}{p}}$. For $p=2$, $\|x\|_{L_2(\sigma)}^2:=\tr(x^*\si^{1/2}x\si^{1/2})$ and we denote the corresponding inner product as $\langle \cdot,\cdot\rangle_\sigma$. We recall from \cite{carlen2017gradient} that the Lindblad operators satisfy $\Delta_\sigma(v_j):=\sigma v_j \sigma^{-1}=e^{-\omega_j}v_j$ where $\Delta_\sigma(x)=\sigma x\sigma^{-1}$ is the modular operator.
The Dirichlet form of $\cL$ is given by
 \begin{align}\mathcal{E}(x):=-\lan x,\cL( x) \ran_{\si}=\sum_{j\in\mathcal{J}}\norm{\partial_j x}{L_2(\si)}^2\pl.\label{eq:de}\end{align}
\begin{theorem}\label{curvatureTI}
Let $ \cP_t:\mathbb{M}_n(\mathbb{C})\to \mathbb{M}_n(\mathbb{C})$ be a $\operatorname{GNS}$-symmetric quantum Markov semigroup with invariant subalgebra $\cN$.
  Assume that the semigroup $(\cP_t)_{t\ge 0}$ satisfies 
  \[\vertiii{\cP_t(x)}_{\partial ,2}\le Ce^{-\kappa t}\,\vertiii{x}_{\partial ,2}\]
  for some $C,\kappa>0$. Then the following transportation information inequality holds: for any $\rho\in\mathcal{D}(\mathbb{M}_n(\mathbb{C}))$ and any faithful invariant state $\sigma$,
 \begin{align}\label{TIineq}
     W_1(\rho,\rho\circ E_{\cN})\le\,\frac{C}{\kappa}\,\max_{i\in\mathcal{J}}\big(e^{-\frac{\omega_i}{4}}+e^{\frac{\omega_i}{4}}\big)\,\sqrt{\cE\big(\Gamma_\sigma^{-\frac{1}{2}}(\sqrt{\rho})\big)}\,.
 \end{align}
\end{theorem}
\begin{proof}
Let $\rho\in\cD(\mathbb{M}_n(\mathbb{C}))$ and take $y=\Gamma_\sigma^{-1}(\rho)$ such that $\rho=\Gamma_\sigma(y)$. Also denote $E^\dagger$ (resp. $\cP_{t}^{\dagger}$) as the adjoint map of $E\equiv E_{\cN}$ (resp. $\cP_t$) with respect to the trace inner product. Then,
\begin{align}
    W_1(\rho,E^\dagger(\rho))&=\sup_{\vertiii{x}_{\partial,2} \le 1}\,\tr\big[(\rho-E^\dagger(\rho))x\big]\nonumber\\
    &=\sup_{\vertiii{ x}_{\partial,2} \le 1}\,-\,\int_0^\infty\,\frac{d}{dt}\tr\big[ \cP_{t}^{\dagger}(\rho)\,x \big]\,dt\nonumber\\
    &=\sup_{\vertiii{x}_{\partial,2} \le 1}\,\int_0^\infty\,\sum_{i\in\mathcal{J}}\langle  \partial_i y,\,\partial_i\cP_t(x)\rangle_\sigma \,dt\nonumber\\
    &=\sup_{\vertiii{x}_{\partial,2} \le 1}\,\int_0^\infty\,\sum_{i\in\mathcal{J}}\norm{\partial_i y}{L_1(\sigma)}
    \|{\partial_i \cP_t(x)}\|
    dt\nonumber\\
    &\le \Big(\sum_{i\in\mathcal{J}}\,\|\partial_i y\|_{L_1(\sigma)}^2\Big)^{\frac{1}{2}}\cdot\sup_{\vertiii{x}_{\partial ,2} \le 1} \,\int_0^\infty\, \Big(\sum_{i}\|{\partial_i \cP_t(x)}\|^2\Big)^{1/2} \,dt\nonumber\\
    &= \Big(\sum_{i\in\mathcal{J}}\,\|\partial_i y\|_{L_1(\sigma)}^2\Big)^{\frac{1}{2}}\cdot\sup_{\vertiii{x}_{\partial ,2} \le 1} \,\int_0^\infty\, \vertiii{\cP_t(x)}_{\partial ,2} \,dt\nonumber\\
    &\le C\,\Big(\sum_{i\in\mathcal{J}}\,\|\partial_i y\|_{L_1(\sigma)}^2\Big)^{\frac{1}{2}}\cdot \,\int_0^\infty\,e^{-\kappa t}\,dt\nonumber\\
    &=\frac{C}{\kappa}\,\Big(\sum_{i\in\mathcal{J}}\,\|\partial_i y\|_{L_1(\sigma)}^2\Big)^{\frac{1}{2}}\,.\label{eq:L1norm}
\end{align}
Next, we control \eqref{eq:L1norm} in terms of the Dirichlet form at $z:= \Gamma_\sigma^{-\frac{1}{2}}(\sqrt{\rho})$. Recall that $\partial_i=[v_i,\cdot]$ and $\Delta_\sigma(v_i)=e^{-\omega_i}v_i$. Write $s_i:=\Delta_\sigma^{\frac{1}{4}}(v_i)$ and $t_i:=\Delta_{\sigma}^{-\frac{1}{4}}(v_i)$. We have
\begin{align}
    \|\partial_i y\|_{L_1(\sigma)}=\|\sigma^{\frac{1}{2}}[v_i,y]\sigma^{\frac{1}{2}}\|_1&=\|\sigma^{\frac{1}{4}}[s_i,z]\sigma^{\frac{1}{2}}z\sigma^{\frac{1}{4}}+\sigma^{\frac{1}{4}}z\sigma^{\frac{1}{2}}[t_i,z]\sigma^{\frac{1}{4}}\|_1\nonumber\\
    &\le \|\Gamma_\sigma^{\frac{1}{2}}(z)\|_2\,\Big(\|\Gamma_\sigma^{\frac{1}{2}}([s_i,z])\|_2\,+\|\Gamma_\sigma^{\frac{1}{2}}([t_i,z])\|_2 \Big)\nonumber\\
    &=\|[s_i,z]\|_{L_2(\sigma)}+\|[t_i,z]\|_{L_2(\sigma)}\nonumber\\
    &=\big(e^{-\frac{\omega_i}{4}}+e^{\frac{\omega_i}{4}}\big)\, \|[v_i,z]\|_{L_2(\sigma)}\,.\label{eqL1L2}
\end{align}
The result follows from \eqref{eq:de} after summing over $i\in\mathcal{J}$.
\end{proof}

Using the result in \Cref{curvatureTI}, we can also derive transportation information inequalities for KMS-symmetric generators. Here we consider the local Lindbladian $\cL$: given a hypergraph $G:=(V,E)$ and the total Hilbert space $\cH_V:=\bigotimes_{v\in V}\cH_v$, $\cH_v\equiv \mathbb{C}^d$ with some $d\in\mathbb{N}$ for all $v\in V$, we consider the generator $\cL_V:=\sum_{e\in E}\cL_e$, where each local generator $\cL_e$ acts non-trivially only on the subsystem $\cH_e:=\bigotimes_{v\in e} \cH_v$, $e\in E$. We assume that each $\cL_e$ is KMS-symmetric with respect to a fixed state $\sigma$ on $\cH_V$ and denote $\displaystyle E_e:=\lim_{t\to\infty}\,e^{t\cL_e}$ and $\displaystyle  E_V:=\lim_{t\to\infty}e^{t\cL_V}$ as the conditional expectations. For each $e\in E$, the map $\widetilde{\cL}_e\equiv E_e-\id$ is the generator of a GNS-symmetric quantum Markov semigroup with respect to the invariant state $\sigma$, and thus it takes the form of Equation \eqref{generatorGNSsymm}, with Bohr frequencies $ {\omega}_{e,j}\in \mathbb{R}$ and Lindblad operators $ v_{e,j}$, $j\in \mathcal{J}_e$.
\begin{proposition}
With the above notations, assume that for each $t$, the channel $\cP_t=e^{t\cL_V}$ satisfies the following gradient estimate with respect to the Ornstein semi-norm: for any $x\in\cB(\cH_V)$,
\begin{align}\label{eq:curvornstein}
    \|e^{t\cL_V}(x)\|_\otimes \,\le C\,e^{-t\kappa}\,\|x\|_\otimes\,,\qquad\kappa,C>0\,.
\end{align}
Assume further that the spectral gaps of the generators $\cL_e$ share a uniform lower bound $\lambda>0$. Then the following transportation information inequality holds: for any $\rho\in\cD(\cH_V)$,
\begin{align}\label{eq:localTI}
    \cT_\otimes(\rho,E_{V*}(\rho))\,\le \frac{C\max_{e\in E}\,C_e\sqrt{|\mathcal{J}_e|}\,|e|\,\|\cL_{e*}\|_\diamond}{\kappa\sqrt{\lambda}}\,\sqrt{|E|\,\cE_{\cL_V}\big(\Gamma_\sigma^{-\frac{1}{2}}(\sqrt{\rho})\big)}\,,
\end{align}
where 
\begin{align}
 C_e:=  2\max_{j\in\mathcal{J}_e}e^{-{\omega}_{e,j}}\|\Delta_\sigma^{\frac{1}{2}}(v_{e,j}^*)\|\big(e^{-\frac{{\omega}_{e,j}}{4}}+e^{\frac{{\omega}_{e,j}}{4}}\big)\,.
\end{align}
\end{proposition}
\begin{proof}
For any $\rho\in \cD(\cH_V)$ and all $x\in\cB(\cH_V)$ with $\|x\|_{\otimes}\le 1$, we denote $y=\Gamma_\sigma^{-1}(\rho)\equiv\sigma^{-\frac{1}{2}}\rho\sigma^{-\frac{1}{2}}$ and $\tau_e(X):=d^{-|e|}\,1_e\otimes\tr_e(X)$. Then we have 
\begin{align}
    \tr\big[(\rho-E_{V*}(\rho))x\big]&=-\int_0^\infty \frac{d}{dt}\,\tr[e^{t\cL_{V*}}(\rho)\,x]\,dt\nonumber\\
    &=-\int_0^\infty\,\tr\big[\rho\,\cL_V\,e^{t\cL_V}(x)\big]\,dt\nonumber\\
    &=-\int_0^\infty\,\sum_{e\in E}\langle y,\,\cL_e\circ e^{t\cL_V}(x)\rangle_\sigma\,dt\nonumber\\
    &=-\int_0^\infty\,\sum_{e\in E}\,\langle y-E_e(y),\,  \cL_e\,(\id-\tau_{e})\,e^{t\cL_V}(x) \rangle_\sigma\,dt\nonumber\\
    &\le \int_0^\infty \sum_{e\in E}\,\|y-E_e(y)\|_{L_1(\sigma)}\,\|\cL_e(\id-\tau_e)\,e^{t\cL_V}(x)\|\,dt\nonumber\\
    & \le \sum_{e\in E}\,\|\cL_{e*}\|_{\diamond}\,|e|\,\|y-E_e(y)\|_{L_1(\sigma)}\,\int_0^\infty\,\|e^{t\cL_V}(x)\|_\otimes\,dt\nonumber\\
    &= \frac{C}{\kappa}\sum_{e\in E}\,\|\cL_{e*}\|_{\diamond}\,|e|\,\|y-E_e(y)\|_{L_1(\sigma)}\,\|x\|_\otimes\,.\label{eqL1W1}
\end{align}
 Next, given $z=\Gamma_\sigma^{-\frac{1}{2}}(\sqrt{\rho})$, the generator $\widetilde{\cL}_e\equiv E_e-\id$ satisfies 
\begin{align}
    \|y-E_e(y)\|_{L_1(\sigma)}&=\|\widetilde{\cL}_e(y)\|_{L_1(\sigma)}\nonumber\\
    &\le \sum_{j\in \mathcal{J}_e}\,e^{-{\omega}_{e,j}}\Big(\|v_{e,j}^*[y,v_{e,j}]\|_{L_1(\sigma)}+\|[v_{e,j}^*,y]v_{e,j}\|_{L_1(\sigma)}\Big)\nonumber\\
    &\le \max_{j\in\mathcal{J}_e}e^{-{\omega}_{e,j}}\,\|\Delta_\sigma^{\frac{1}{2}}(v_{e,j}^*)\|\,\,\sum_{j\in\mathcal{J}_e}\,\|[v_{e,j},y]\|_{L_1(\sigma)}\nonumber\\
    &\overset{(1)}{\le}\,\underbrace{\max_{j\in\mathcal{J}_e}e^{-{\omega}_{e,j}}\,\|\Delta_\sigma^{\frac{1}{2}}(v_{e,j}^*)\|\big(e^{-\frac{{\omega}_{e,j}}{4}}+e^{\frac{{\omega}_{e,j}}{4}}\big)}_{C_e}\,\sum_{j\in\mathcal{J}_e}\,\|[v_{e,j},z]\|_{L_2(\sigma)}\nonumber\\
    &\le C_e\,\sqrt{|\mathcal{J}_e|\,\mathcal{E}_{E_e-\id}(z)}=C_e \sqrt{|\mathcal{J}_e|\,\|z-E_e(z)\|_{L_2(\sigma)}^2}\overset{(2)}{\le} \frac{C_e}{\sqrt{\lambda}}\,\sqrt{|\mathcal{J}_e|\,\cE_{\cL_e}(z)}\,.\label{eqW1L1}
\end{align}
In (1) above, we use the fact that the generator $E_e-\id$ is GNS-symmetric, so that the same argument as for \eqref{eqL1L2} is applied. (2) follows from the assumption of uniform lower bound $\lambda$ on the local spectral gaps of the generators $\cL_e$. Using \eqref{eqW1L1} into \eqref{eqL1W1}, we get the bound
\begin{align*}
    \tr\big[(\rho-E_{V*}(\rho))\,x\big]&\le \frac{C\max_{e\in E}\,C_e\sqrt{|\mathcal{J}_e|}\,|e|\,\|\cL_{e*}\|_\diamond}{\kappa\sqrt{\lambda}}\,\sum_{e\in E}\,\sqrt{\mathcal{E}_{\cL_e}(z)}\\
    &\le \frac{C\max_{e\in E}\,C_e\sqrt{|\mathcal{J}_e|}\,|e|\,\|\cL_{e*}\|_\diamond}{\kappa\sqrt{\lambda}}\,\Big(|E|\,\sum_{e\in E}\,{\mathcal{E}_{\cL_e}(z)}\Big)^{\frac{1}{2}}\\
    &=\frac{C\max_{e\in E}\,C_e\sqrt{|\mathcal{J}_e|}\,|e|\,\|\cL_{e*}\|_\diamond}{\kappa\sqrt{\lambda}}\,\sqrt{|E|\,\cE_{\cL_V}(z)}\,.
\end{align*}
\end{proof}

\begin{rem}
Transportation information inequalities of the form of \eqref{TIineq} and \eqref{eq:localTI} were recently used in \cite{BenoistHanggliRouze2021} to obtain concentration inequalities for quantum trajectories.
\end{rem}

\subsection{Diameter estimates}
Let $(\cM,\cB)$ be a non-commutative transportation metric. In analogy with \cite{Ollivier2009}, we define the \textit{jump} of a state $\omega\in\cD(\cM)$ under a quantum channel $\cP$ as
\begin{align}\label{jump}
    \mathcal{J}_{\cB}(\omega):=\cT_\cB(\omega,\omega\circ \cP)\,.
\end{align}

\begin{lemma}\label{distancetojump}
Assume that $\cT_\cB$ satisfies the triangle inequality and that the triple $(\cM,\cB,\cP)$ has positive coarse Ricci curvature bound $\kappa>0$. Then, for any $\omega_1,\omega_2\in\cD(\cM)$, 
\begin{align}
    \cT_\cB(\omega_1,\omega_2)\le \frac{1}{\kappa}\big(\mathcal{J}(\omega_1)+\mathcal{J}(\omega_2)\Big)\,.
\end{align}
\end{lemma}
\begin{proof}
By triangle inequality of $\cT_{\cB}$ and curvature of $(\cM,\cB,\cP)$, 
\begin{align*}
\cT_\cB(\omega_1,\omega_2)&\le \cT_\cB(\omega_1,\omega_1\circ \cP)+\cT_\cB(\omega_1\circ \cP,\omega_2\circ \cP)+\cT_\cB(\omega_2\circ \cP,\omega_2\circ \cP)
\\ &\le \cT_\cB(\omega_1,\omega_1\circ \cP)+(1-\kappa)\cT_\cB(\omega_1,\omega_2)+\cT_\cB(\omega_2\circ \cP,\omega_2\circ \cP)
\end{align*}
 Rearranging the terms gives the assertion.
\end{proof}

\subsection{Existence of the coarse curvature bound}
In this section, we prove that in finite dimensions, the Lipschitz constant always exists for the semi-norm given by a family of commutators.

\begin{proposition}
Let $\cH$ be a finite dimensional Hilbert space and $(\cB(\cH),\cB_L)$ be the non-commutative transportation metric induced by the semi-norm \[L(x):=\max_{j\in\mathcal{J}}\|[A_j,x]\|\] for some finite family $\{A_j\}_{j\in \mathcal{J}}\in \cB(\cH)$.
Denote the kernel of the semi-norm $L$ as $\cN=\{A_1,\dots,A_{|\mathcal{J}|}\}'$.
Then for any quantum channel $\cP:\cB(\cH)\to \cB(\cH)$ such that $\cP(\cN)\subset \cN$, there exists $\kappa\in \mathbb{R}$ such that $\cP$ satisfies the $\kappa$-Lipschitz estimate and hence
$(\cB(\cH),\cB_L,\cP)$ has coarse Ricci curvature bounded by $\kappa$.
\end{proposition}
\begin{proof}Denote the gradient
\[ \partial: \cB(\cH)\to \bigoplus_{j\in\mathcal{J}} \cB(\cH)\pl ,\pl \quad   \partial(x)=([A_j,x])_{j\in \mathcal{J}}\pl.\]
We have $L(x)=\|{\partial(x)}\|$. Let $\cN^\perp$ be orthogonal complement of $\cN$ in $\cB(\cH)$ with respect to the trace inner product. Since both $[A_j,x]$ and $[A_j,\cP(x)]$ vanish on $\cN$ for any $j\in\mathcal{J}$, it suffices to prove the existence of $\kappa\in\mathbb{R}$ such that for any $x\in\cN^\perp$, 
\[\|{\partial(\cP(x))}\|\le (1-\kappa) \|{\partial (x)}\|\,. \]
Consider the Stinespring dilation of $\cP$
\begin{align*}
    \cP(x)=\tr_{\cK}\big((1\otimes\sigma)\,U^*(x\otimes 1)U\big)\,,
\end{align*}
where $U:\cH\otimes \cK\to \cH\otimes \cK$ is a unitary, $\tr_{\cK}$ is the partial trace on $\cK$ and $\si$ is a density operator on $\cK$. We have
\begin{align*}
    \big[A_j,\cP(x)\big]=\tr_\cK\big( (1\otimes\sigma)[A_j\otimes1,U^*]x \,U\big)+\cP\big([A_j,x]\big)+\tr_\cK\big((1\otimes \sigma)\,U^*(x\otimes 1)[A_j\otimes 1,U]\big)\,.
\end{align*}
Define $B_j:=[A_j\otimes1,U]\,U^*$, so that $[A_j\otimes1,U]=B_jU$ and $
[A_j\otimes 1,U^*]=-U^*B_j$. 
Then,
\begin{align*}
    \big[A_j,\cP(x)\big]=\id\ten \si\Big( U^*[A_j\otimes 1-B_j, x\ten 1]U\Big)\,.
\end{align*}
Note that $\text{Ran}(\partial)=\partial(\cN^\perp)$ and $\partial$ is bijective from $\cN^\perp$ to $\text{Ran}(\partial)$ since $\cH$ is finite dimensional. Take the inverse $\partial^{-1}:\text{Ran}(\partial)\to\cN^\perp$ and
define the map $\tilde{\partial}(x)=([A_j\otimes 1-B_j, x\ten 1])_{j\in\mathcal{J}}$. Therefore,
\begin{align*}
    \|[A_j,\cP(x)]\|\le \|{ \tilde{\partial}(x)}\|
    \,\le 
    \,\|{\tilde{\partial}\circ \partial^{-1}:\operatorname{Ran}(\partial)\to\bigoplus_{j\in\mathcal{J}}\cB(\cH\otimes \cK)}\|\, \|{\partial(x)}\|\,.
\end{align*}
Again, in finite dimensions, the operator norm of $\tilde{\partial}\circ \partial^{-1}$ from $\operatorname{Ran}(\partial)$ to $\bigoplus_{j\in \mathcal{J}}\cB(\cH\otimes \cK)$ is finite. The result follows.
\end{proof}

\begin{rem}In the above proposition, the assumption $\cP(\cN)\subset \cN$ is necessary. Indeed, if there exists $a\in\cN$ such that $\cP(a)\notin \cN$, then
\[ 0<\|{\partial{\cP(a)}}\|\le (1-\kappa)\|{\partial{a}}\|=0,\]
which implies that $\kappa$ cannot be finite.
\end{rem}


\section{Gradient estimate via Intertwining and transference}
\subsection{Intertwining relation}\label{sec:intertwining}
In this section, we provide two approaches to derive gradient estimates for a quantum channel. The first one is the intertwining relation 
which was found as a useful tool to derive curvature conditions for quantum Markov semigroups in \cite{carlen2017gradient} (see also \cite{wirth2020complete, brannan2020complete,datta2020relating}). 

Let $(\cM,\tau)$ be a finite von Neumann algebra and let 
$(\cA,\cH, l,r, \partial)$ be a first order structure as in section \ref{sec:differential}. 
Let $\cP:\cM\to\cM$ be a quantum channel. We are interested in the intertwining relation $\partial\cP= \hat{\cP}\partial$ for some linear map $\hat{\cP}:\cH\to \cH$.
We first consider the non-commutative transportation metric induced by the Lipschitz semi-norm
\[L:\cA\to [0,\infty), \qquad L(x)=\norm{\partial x}{L}\]
where $\norm{\cdot}{L}$ is some norm on $\partial(\cA)$, which can be different from the Hilbert space norm. We denote $\cH_L$ as the completion of $\partial(\cA)$ under $\norm{\cdot}{L}$. By Proposition \ref{curvaturetolip}, we know
\[L(\cP(x))\le (1-\kappa) L(x) \pl \forall\pl x\in \cA \Longrightarrow 
W_L(\omega_1\circ \cP,\omega_2\circ \cP)\le (1-\kappa)W_L(\omega_1,\omega_2)\pl \forall\pl \omega_1,\omega_2 \in \cD(\cM)\pl.
\]
One immediately has the following proposition:
\begin{proposition}\label{prop:L1inter}
Let $\cP:\cM\to\cM$ be a quantum channel and suppose $\partial\cP= \hat{\cP}\partial$ for some linear map $\hat{\cP}:\cH\to \cH$. Then 
\[ L(\cP(x))\le \norm{\hat{\cP}:\cH_L\to \cH_L}{} L(x)\,.\]
In particular,
if $\cH=L_2(\hat{\cM})$ for some finite von Neumann algebra $\hat{\cM}$ as in Equation \eqref{eq:algebra} and
$\hat{\cP}$ is a completely positive and trace-symmetric map, then for $1\le  p\le \infty$,
    \[ \norm{\cP(x)}{\partial,p}\,\le\, \|\hat{\cP}(1)\|{}\,\norm{\cP(x)}{\partial,p}\]
where $\norm{x}{\partial,p}:=\norm{\partial x}{L_p(\cM)}$.
\end{proposition}
\begin{proof}
The first assertion is straightforward from the intertwining relation. For the second assertion, it suffices to 
show that $\norm{\hat{\cP}:L_p(\cM)\to L_p(\cM)}{\operatorname{cb}}\le \|\hat{\cP}(1)\|{}$. This is known for $p=\infty$ and by symmetry for $p=1$. The case $1<p<\infty$ follows by interpolation.
\end{proof}

For the quantum Wasserstein 2-distance defined in \eqref{eq:Riemmannian}, we have the following analog of \cite[Theorem 3.1]{wirth2020complete} for a quantum channel.
\begin{theorem}\label{thm:L2inter}
Suppose there exists a map $\hat{\cP}:\cH\to \cH$ and $C\ge 0$ such that
\begin{enumerate}
    \item[$\operatorname{(i)}$] $\partial\cP=\hat{\cP}\partial$;
 \item[$\operatorname{(ii)}$] $\hat{\cP}^\dagger \circ l(\rho )\circ \hat{\cP}\le C \,l(\cP^\dagger( \rho))$ as an operator in $\cB(\cH)$ for any $\rho\in \cM_+$;
 \item[$\operatorname{(iii)}$] $\hat{\cP}^\dagger \circ r(\rho )\circ \hat{\cP}\le C\, r(\cP^\dagger (\rho))$ as an operator in $\cB(\cH)$ for any $\rho\in \cM_+$ , 
\end{enumerate}
where $\hat{\cP}^\dagger$ is the adjoint map of $\hat{\cP}$ on $\cH$. 
Then $\cP$ satisfies $(1-C)$-$\operatorname{CGE}$ and for any $\omega_1, \omega_2\in \cD(\cM)$,
\begin{align}
W_{\partial,2}(\omega_1\circ \cP , \omega_2\circ \cP)\le C\, W_{\partial,2}(\omega_1, \omega_2)\,.\label{eq:W2contraction}
\end{align}
\end{theorem}
\begin{proof}
Let $\Lambda_\rho=\Lambda(l(\rho),r(\rho))$ be the operator mean of $l(\rho)$ and $r(\rho)$ in the definition \eqref{metricCM} of the metric. Then by the defining property of operator means,
\begin{align}\label{eq:order} \hat{\cP}^\dagger\circ  \Lambda_\rho\circ  \hat{\cP}&=\hat{\cP}^\dagger\circ \Lambda(l(\rho),r(\rho))\circ \hat{\cP}
\nonumber\\ &\le \Lambda(\hat{\cP}^\dagger \circ l(\rho)\circ \hat{\cP},\hat{\cP}^\dagger \circ r(\rho)\circ  \hat{\cP})
\nonumber\\ &\le \Lambda(C\,l(\cP^\dagger(\rho)),\,C \,r(\cP^\dagger(\rho)) )
\nonumber \\ &=C \,\Lambda(l(\cP^\dagger(\rho))\hat{\cP}, r(\cP^\dagger(\rho)) )=C\,\Lambda_{\cP^\dagger(\rho)}
\end{align}
as a positive operator on $\cH$. Therefore,
\begin{align*}
\norm{\partial\cP(x)}{\rho}^2=&\norm{\hat{\cP}(\partial x)}{\rho}^2
\\= &
\lan \hat{\cP}(\partial x),\, \Lambda_{\rho}(\hat{\cP}(\partial x)) \ran_{\cH}
\\= &
\lan\partial x,  \hat{\cP}^\dagger\circ \Lambda_{\rho}\circ \hat{\cP}(\partial x) \ran_{\cH}
\\ \le & C\lan\partial x,  \Lambda_{\cP^\dagger(\rho)}(\partial x) \ran_{\cH}=\norm{\partial x}{\cP^\dagger(\rho)}\pl.
\end{align*}
Applying the same argument to $\id_{\mathcal{R}}\ten \partial$ for all finite von Neumann algebra $\mathcal{R}$ gives $(1-C)-\operatorname{CGE}$. The assertion for the Wasserstein distance follows from Proposition \ref{prop:equivalence}. 
\end{proof}

\subsection{Group transference}\label{sec:transference}
In this part, we show that the gradient estimate of a quantum channel $\cP$ can be transferred from a classical Markov map. Let $G$ be a compact group equipped with the normalized Haar measure $\mu$.
Consider a left invariant Markov map
\[K: L_\infty(G)\to L_\infty(G)\pl, \qquad \pl (K f)(h):=\int_{G} k(h^{-1}g)f(g)\,d\mu(g)  \]
where  $\mu$ is the normalized Haar measure on $G$ and the kernel function is $k(g,h):=k(h^{-1}g)$. Given a projective unitary representation $u:G\to \mathcal{U}(\cH)$ of $G$ on some finite dimensional Hilbert space $\cH$, we define the transference map
\[ \pi: \cB(\cH)\to L_\infty(G,\cB(\cH))\pl,\qquad \pl \pi(x)(g)=u(g)xu(g)^*\pl.\] The tranferred quantum channel on $\cB(\cH)$ is given by 
\begin{align}\label{transferred}
 \cP:\cB(\cH)\to \cB(\cH)\pl, \pl    \qquad  \cP(x):=\int_G\,k(g)\,u(g) \,x\,u(g)^*\,d\mu(g)\,,
\end{align}
which satisfies the following commuting diagram
\begin{equation}\label{ccss}
 \begin{array}{ccc} L_{\infty}(G,\cB(\cH))\pl\pl &\overset{K\ten \id_{\cB(\cH)}}{\xrightarrow{\hspace*{1.6cm}}} & L_{\infty}(G,\cB(\cH)) \\
                    \uparrow  \pi   & & \uparrow \pi \\
                     \cB(\cH) \pl\pl &\overset{\cP}{\longrightarrow} & \cB(\cH)
                     \end{array} \pl .
                     \end{equation}
Since $\pi$ is a trace preserving $*$-homomorphism, one can view $\cP$ a the reduced action of $K\ten \id_{\cB(\cH)}$ on the sub-system $\cB(\cH)$. The transference technique passes many properties from the classical map $K$ to the quantum channel $\cP$. (see \cite{gao2020fisher,bardet2019group} for similar discussions in the case of a semigroup). 

Here we first consider gradient estimates induced by a differential structure.
Let $\mathfrak{g}$ be the Lie algebra of left invariant vector fields, i.e. for $X\in \mathfrak{g}$. 
\[Xf(g):=\frac{d}{dt}f(ge^{tX})|_{t=0}\pl.\]
Given a left invariant metric on $G$ with corresponding geodesic distance $d_{\operatorname{Riem}}$, the gradient operator can be defined as
\[\nabla: C^\infty(G)\to \oplus_{j=1}^d C^\infty(G)\pl, \qquad \pl \nabla(f)=(X_jf)_{j=1}^d\pl,\]
where $\{X_j\}_{j=1}^d\subset \mathfrak{g}$ is an o.n.b. with respect to the metric and $C^\infty(G)$ denotes the class of smooth functions on $G$. We choose $L_j\in \cB(\cH)_\mathbb{R}$ to be the self-adjoint operators satisfying the relations $e^{itL_j}=u(e^{tX_j})$. In words, $L_j$ is the image of $X_j$ under the projective Lie algebra representation (up to the imaginary unit $i$). We define the corresponding non-commutative differential structure on $\cB(\cH)$ as follows
\[\partial=(\partial_j):\cB(\cH)\to \bigoplus_{j=1}^d \cB(\cH)\pl ,\pl\qquad  \partial_j(x)=[L_j,x] \pl.\]
The following interwining relation holds: for any $x\in \cB(\cH)$,
\begin{align}\nonumber
\big( X_j\pi(x)\big)(g)&=\left.\frac{d}{dt}\right|_{t=0} u(ge^{tX_j})xu(ge^{tX_j})^*\\
&=\left.\frac{d}{dt}\right|_{t=0} u(g) e^{iL_jt}xe^{-iL_jt}u(g)^*\nonumber
\\
&=u(g) [iL_j,x]u(g)^*\nonumber
\\
&=i\pi(\partial_j x)(g) \pl. \label{eq:inter}
\end{align}

In the next theorem, we show that the above intertwining relation enables us to pass gradient estimates for $K$ to its transferred map $\cP$.
\begin{theorem}\label{transferrencetheorem}
Let $G$ be a compact Lie group and let $K$ be a left invariant Markov map on $L_\infty(G)$ with invariant (Haar) measure $\mu$. Let $u:G\to \mathcal{U}(\cH)$ be a continuous projective unitary representation on a (finite dimensional) Hilbert space $\cH$ and $\cP$ be the transferred quantum channel defined above.
\begin{enumerate}
 \item[$\operatorname{(i)}$]Suppose $K:L_\infty(G)\to L_\infty(G)$ satisfies the following pointwise gradient estimate: for any $f\in\operatorname{dom}(\nabla)$, 
\begin{align}\label{eq:Ricciclass}
|\nabla K(f)|^2\le (1-\kappa)\,K(| \nabla f|^2)\,
\end{align}
where $|\nabla f|^2=\sum_{j=1}^d|X_j f|^2$.
Then for $s\in[0,1]$ and any positive operators $A,B$,
\begin{align} \label{eq:gradestimate}\langle \partial \cP(x), \partial \cP(x)\rangle_{A,B,s}\le  (1-\kappa)\,\langle \partial x, \partial x\rangle_{\cP^\dag(A),\cP^\dag(B),s}  \end{align}
where $\cP^\dag$ is the adjoint channel of $\cP$ w.r.t. the trace inner product, and where the inner product $\langle.,.\rangle_{A,B,s}$ is defined for any two vectors ${\bf x}=(x_j),{\bf y}=(y_j)\in \bigoplus_{j=1}^d{\cB(\cH)}$ as
\[ \langle {\bf x}, {\bf y} \rangle_{A,B,s}:=\sum_{j=1}^d\tr\big(x_j^*A^{s}y_jB^{1-s}\big)\pl.\]
\item[$\operatorname{(ii)}$]For $1\le p\le\infty$, define the following weak $L_2(L_p)$ norm: for ${\bf x}=(x_j)\in \bigoplus_{j=1}^d{\cB(\cH)}$
\[\norm{{\bf x}}{2,p}=\sup_{\|{A }\|_{p'}\le 1} \big(\sum_{j=1}^d |\tr(Ax_j)|^2\big)^{1/2}\pl,\]
where $1/p+1/p'=1$.
Suppose that $K:L_\infty(G)\to L_\infty(G)$ satisfies the following gradient estimate: for any $f\in\operatorname{dom}(\nabla)$, 
\begin{align}\label{eq:Lip}
\norm{K(f)}{\operatorname{Lip}}\le (1-\kappa)\norm{f}{\operatorname{Lip}}
\end{align}
where $\norm{f}{\operatorname{Lip}}=\displaystyle \sup_{g\ne h\in G}\frac{|f(g)-f(h)|}{d_{\operatorname{Riem}}(g,h)}=\sup_{g\in G} |\nabla f|(g)$ is the Lipschitz constant. Then for all $1\le p\le \infty$,
\[\norm{{\partial \cP(x)}}{2,p}\le (1-\kappa) \norm{{\partial x}}{2,p}\pl.\]
\end{enumerate}
\end{theorem}

\begin{proof}
(i) For sake of simplicity, we denote $K\equiv K\otimes \id_{\cB(\cH)}$. For each $j$, we have by \eqref{eq:inter} and $\pi\cP=K\pi$,
\begin{align}
i\partial_j\cP(x)=i\pi(\partial_j\cP(x))(e)=X_j(\pi\cP(x))(e)=X_j(K\pi(x))(e)\,.\label{commutationrelations}
\end{align}
Therefore, for any two vectors $\varphi,\psi\in\cH$,
\begin{align*}
\langle \varphi|\partial_j\cP(x)|\psi\rangle = -i\langle \varphi|X_j(K\pi(x))(e)|\psi\rangle= -iX_jK (\langle \varphi|\pi(x)|\psi\rangle) (e)= -iX_jK (f_{\varphi,\psi}^x) (e)\,,
\end{align*}
where $f^x_{\varphi,\psi}(g):= \langle \varphi|u(g) x u(g)^*|\psi\rangle$. 
Fix $s\in[0,1]$. For any two positive matrices $A:=\sum_k\lambda_k|\varphi_k\rangle\langle \varphi_k|$ and $B:=\sum_{l}\mu_l|\psi_l\rangle\langle \psi_l|$ with $\la_k,\mu_l>0$, we consider the trace 
\begin{align}\label{eq:lhs}
\tr\big( \partial_j\cP(x)^*A^s\,\partial_j\cP(x)B^{1-s}\big)=\sum_{k,l}\lambda_k^s\mu_l^{1-s}\langle \psi_l|\partial_j\cP(x)^*|\varphi_k\rangle \langle \varphi_k|\partial_j\cP(x)|\psi_l\rangle\,.
\end{align}
 This together with the gradient estimate \eqref{eq:Ricciclass} gives
\begin{align}
\sum_j	\tr\big( \partial_j\cP(x)^*A^s\,\partial_j\cP(x)B^{1-s}\big)&=\sum_{k,l}\Big(\sum_{j}\lambda_k^s\mu_l^{1-s}|X_j K\,(  f^x_{\varphi_k,\psi_l})(e)|^2\Big)\nonumber
\\&=\sum_{k,l}\lambda_k^s\mu_l^{1-s}|\nabla K\,(  f^x_{\varphi_k,\psi_l})(e)|^2
\nonumber\\
&\overset{(1)}{\le }(1-\kappa)\sum_{k,l}\lambda_k^s\mu_l^{1-s}\,K(|\nabla  f^x_{\varphi_k,\psi_l}|^2)(e)\,,\label{intermediatestep}
\end{align}
where (1) is a consequence of the pointwise gradient estimate \eqref{eq:Ricciclass}.
Using \eqref{eq:inter},
\begin{align*}
  X_jf^x_{\varphi,\psi}(g)= \langle\varphi|X_j\pi(x)  |\psi\rangle(g)=  i\langle\varphi|\pi(\partial_j x)  |\psi\rangle(g)
 =  i\langle\varphi|u(g)(\partial_j x)u(g)^*  |\psi\rangle
\end{align*}
Hence, using the above identity in the right-hand side of \eqref{intermediatestep}, we get:
	\begin{align}
		\sum_{k,l}\lambda_k^s\mu_l^{1-s}\,&K(|\nabla f^x_{\varphi_k,\psi_l}|^2)(e)\nonumber\\
		&=\int_G\,\sum_{j,k,l}\lambda_k^s\mu_l^{1-s}\,k(g)  \,    | \langle \varphi_k|u(g) ( \partial_j  x)u(g)^* |\psi_l\rangle|^2   \, d\mu(g)\nonumber\\
		&=\int_G\,k(g)\,\sum_{j}\,   \tr\big((\partial_j x)^*\,(A(g))^s\,(\partial_j x)\,(B(g))^{1-s} \big)  \, d\mu(g)\nonumber\\
		&\overset{(1)}{\le} \sum_{j}\,   \tr\big((\partial_j x)^*\,(\cP^\dag(A))^s\,(\partial_j x)\,(\cP^\dag(B))^{1-s} \big) \,,\label{eq:rhs}
	\end{align}
	where $A(g):=u(g)^* A\, u(g), B(g)=u(g)^* B\, u(g)$ and $\cP^\dag(x)=\int_{G} k(g)u(g)^*xu(g)d\mu(g)$ is the adjoint channel of $\cP$.
Here, the identity (1) follows from Lieb's concavity theorem \cite{Lieb1973}. The result follows.  
\medskip

(ii) Given $x,A\in \cB(\cH)$, we denote the function $f_A^x(g):=\tr(u(g)xu(g^*)A)$.
We have 
\begin{align*} \norm{\partial x}{2,p}^2=\sup_{\|A\|_{p'}\le 1} 
\sum_j|\Tr((\partial_j x)A)|^2
=&\sup_{\|A\|_{p'}\le 1}\sum_j\left|\Tr\big( \left(X_j\pi(x)(e)\right) A\big)\right|^2
\\ =& \sup_{\|A\|_{p'}\le 1}\sum_j\left|X_j f_A^x(e)\right|^2
\\ \overset{(1)}{=}& \sup_{\|A\|_{p'}\le 1}\sup_{g\in G}\sum_j\left|X_j f_A^x(g)\right|^2
\\ =&  \sup_{\|A\|_{p'}\le 1}\norm{ f_A^x}{\text{Lip}}^2
\end{align*}
In the inequality (1), we used that for any $g\in G$,
\begin{align} X_j f_A^y(g)=X_j\Tr(Au(g)yu(g)^*)&=\left.\frac{d}{dt}\right|_{t=0} \Tr(Au(ge^{tX_j})yu(ge^{tX_j})^*)\nonumber
\\&=\left.\frac{d}{dt}\right|_{t=0} \Tr(u(g)^*Au(g)u(e^{tX_j})yu(e^{tX_j})^*)\nonumber
\\&=X_{j}f_{u(g)^*Au(g)}^y(e) \label{trickyequation}
\end{align}
Therefore, using once again the commutation relation \eqref{commutationrelations},  
\begin{align*} \norm{\partial\cP(x)}{2,p}^2=\sup_{\|A\|_{p'}\le 1} 
\sum_j|\Tr(\partial_j\cP(x)A)|^2
&=\sup_{\|A\|_{p'}\le 1}\sum_j\left|\Tr\big( \left(X_j(K\pi(x))(e)\right) A\big)\right|^2
\\ &= \sup_{\|A\|_{p'}\le 1}\sum_j\left|X_j Kf_A^x(e)\right|^2
\\ &\overset{(2)}{=} \sup_{\|A\|_{p'}\le 1}\,\sup_{g\in G}\sum_j\left|X_j Kf_A^x(g)\right|^2
\\ &=  \sup_{\|A\|_{p'}\le 1}\norm{ Kf_A^x}{\text{Lip}}^2
\\ &\le  (1-\kappa)^2\sup_{\|A\|_{p'}\le 1}\norm{ f_A^x}{\text{Lip}}^2=\, \norm{\partial x}{2,p}^2\,,
\end{align*}
where (2) above follows once again from \eqref{trickyequation}. The result follows. 
\end{proof}

The contraction of the Lipschitz constant used in \Cref{transferrencetheorem}(ii) is equivalent to the coarse Ricci curvature bound introduced by Ollivier \cite{Ollivier2009}. On the other hand, the assumption in \Cref{transferrencetheorem}(i), namely that
\[|\nabla K f|^2\le (1-\kappa)K(|\nabla  f|^2)\]
reduces to Bakry-Emery's curvature dimension condition \cite{bakry1984hypercontractivite,bakry1994hypercontractivite} in the case of a Markov semigroup $K=e^{tL}$ and with $\kappa=1-e^{\kappa' t}$, where $\kappa'$ is a uniform lower bound on the Ricci curvature of the underlying Lie group $G$. We proved that it implies various non-commutative gradient estimates in \eqref{eq:gradestimate}. For example, choosing $s=0$, $A=1_\cH$ and $B$ to be any state $\sigma$, \eqref{eq:gradestimate} gives
	\begin{align}
	   \tr\Big( \sum_j \,|\partial_j\cP(x)|^2\,\sigma\Big)\le (1-\kappa)\tr\Big( \sum_j \,|\partial_j x|^2\,\cP^{\dag}(\sigma)\Big)= (1-\kappa)\tr\Big(\cP\big(\sum_j \,|\partial_j x|^2\big)\,\sigma\Big).
	\end{align}
Since $\sigma$ is arbitrary, we get the following (completely positive) Bakry-Emery type condition:
\[ \Gamma(\cP(x))\le_{\operatorname{cp}} (1-\kappa) \cP \Gamma(x)\pl\]
for the gradient form $\Gamma(x)=\sum_{j}|\partial_j x|^2$. Note that the complete positivity follows from applying the same estimate \eqref{eq:gradestimate} to the representation $\tilde{u}(g)=u(g)\ten 1_{\cK}$ for any Hilbert space $\cK$.
On the other hand, choosing $A=B=\rho$ yields the bound:
\begin{align*}
	\|\partial\cP(x)\|_{\rho}^2\le (1-\kappa)\,\|\partial x\|_{\cP^\dag(\rho)}^2\,,
	\end{align*}
	for the inner product \[ \lan\mathbf{A},\mathbf{B}\ran_{\rho}=\sum_{j=1}^d\,\int_0^1\,  e^{\omega_j(s-1/2)}\tr(A_j^*\,\rho^s\,B_j\,\rho^{1-s})ds\pl , \pl \mathbf{A}=(A_j),\mathbf{B}=(B_j)\in\bigoplus_{j
=1}^d\cB(\cH)\]
where $\omega_j$ is the some frequency constant. 
As discussed in	Example \ref{exam:W2}, the above inner product is at the heart of the definition the non-commutative $W_2$ distance introduced by Carlen and Maas \cite{carlen2017gradient} and the associated entropic Ricci curvature lower bound.

The above transference technique applies similarly to finite groups with discrete differential structure, which we briefly illustrate here. Let $G$ be a finite group equipped with the normalized counting measure. Consider a Markov map  
\[ K: \ell_\infty(G)\to \ell_\infty(G)\,,\qquad (K f)(h):=\frac{1}{|G|}\sum_{g} k(h^{-1}g)f(g) \,,\]
where $k\ge 0$ is the kernel function. 
Given a finite dimensional projective unitary representation, the transferred quantum channel on $\cB(\cH)$ is given by
\begin{align}\label{transferred2}
 \cP:\cB(\cH)\to \cB(\cH)\pl, \pl    \qquad  \cP(x):=\frac{1}{|G|}\sum_{g}k(g)\,u(g) \,x\,u(g)^*\,
\end{align}
and satisfies the commutation relation $\pi\cP=(K\ten \id_{\cB(\cH)})\pi$ with the transference map $\pi:\cB(\cH)\to \ell_\infty(G,\cB(\cH))$ defined as in 
 the commuting diagram \eqref{ccss}.
The Markov map $K$ is ergodic (i.e. has a unique invariant state) if the support of $k$ is a generating set. If in addition $k(g)=k(g^{-1})$, then both $K$ and $\cP$ are self-adjoint \cite{bardet2019group}.

We shall now consider a discrete differential structure. Let $S$ be a subset of $G$ and a weight function $w:S\to (0,\infty)$. Define the difference operator
\[\nabla=(\nabla_g f)_{g\in S}: \ell_\infty(G)\to \bigoplus_{g\in S} \ell_\infty(G)\pl, \qquad \pl (\nabla_g f)(h)=\sqrt{w(g)}\,(f(hg)-f(h))\pl .\]
Each $\nabla_g$ is the (left invariant) discrete difference operator via the transition of the edges $\{(h,gh);h\in G\}$ in the Cayley graph, i.e. $[\nabla_g, L_{g'}]=0$ for any operator $L_{g'}$ of left multiplication by $g'\in G$. The associated gradient form is defined as
\[ \Gamma_\nabla(f,f)=\sum_{g\in S} |\nabla_g f|^2\pl,\pl \qquad \Gamma_\nabla(f,f)(h)=\sum_{g\in S}\omega(g)\,|f(hg)-f(h)|^2\,.\]
The corresponding non-commutative differential structure on $\cB(\cH)$ is defined as follows
\[\partial=(\partial_g)_{g\in S}:\cB(\cH)\to \bigoplus_{g\in S} \cB(\cH)\pl ,\pl\qquad  \partial_g(x)=\sqrt{\omega(g)}\,(u(g)xu(g)^*-x)=\sqrt{\omega(g)}\,u(g)[x,u(g)^*] \pl,\]
and the associated non-commutative gradient form is
\[ \Gamma_\partial(x,x)=\sum_{g\in S} |\partial_g x|^2=\sum_{g\in S}\omega(g)\,|[x,u(g)^*]|^2\,.\]
Note that here $\partial$ is not a derivation in the sense that it does not satisfy Leibniz rule.
The following interwining relation holds: for any $x\in \cB(\cH)$,
\begin{align}\nonumber
\big( \nabla_h\pi(x)\big)(g)&=\sqrt{\omega(h)}(u(gh)x u(gh)^*-u(g)xu(g)^*)\\
&= \sqrt{\omega(h)}u(g) (u(h)xu(h^*)-x)u(g)^*\nonumber
\\
&=u(g)(\partial_h x)u(g)^*\nonumber
\\
&=\pi(\partial_h x)(g) \pl. \label{eq:inter2}
\end{align}

\begin{theorem}
Let $G$ be a finite group and let $K$ be a left invariant Markov map on $\ell_\infty(G)$. Let $u:G\to \mathcal{U}(\cH)$ be a projective unitary representation on a (finite dimensional) Hilbert space $\cH$ and $\cP$ be the transferred quantum channel defined above.
\begin{enumerate}
 \item[$\operatorname{(i)}$]Suppose $K:\ell_\infty(G)\to \ell_\infty(G)$ satisfies the following  gradient estimate: for any $f\in \ell_\infty(G)$,
\begin{align}\label{eq:Ricciclass2}
|\nabla K(f)|^2\le (1-\kappa)\,K(| \nabla f|^2)\,
\end{align}
where $|\nabla f|^2=\sum_{g\in S}|\nabla_g f|^2$.
Then for $s\in[0,1]$ and any positive operators $A,B\in\cB(\cH)$,
\begin{align} \label{eq:gradestimate}\langle \partial \cP(x), \partial \cP(x)\rangle_{A,B,s}\le  (1-\kappa)\,\langle \partial x, \partial x\rangle_{\cP^\dag(A),\cP^\dag(B),s}  \end{align}
where $\cP^\dag$ is the adjoint channel of $\cP$ w.r.t. the trace inner product and the inner product $\langle.,.\rangle_{A,B,s}$ is defined as in Theorem \ref{transferrencetheorem}.
\item[$\operatorname{(ii)}$]For $1\le p\le\infty$, define the following weak $L_\infty(L_p)$ norm: for ${\bf x}=(x_g)\in \bigoplus_{g\in S}{\cB(\cH)}$,
\[\norm{{\bf x}}{\infty,p}=\sup_{\|{A }\|_{p'}\le 1}\, \sup_{h\in S}\pl |\tr(Ax_h)|\pl,\]
where $1/p+1/p'=1$.
Suppose that $K:\ell_\infty(G)\to \ell_\infty(G)$ satisfies the following gradient estimate: for any $f\in \ell_\infty(G)$, 
\begin{align}\label{eq:Lip}
\norm{K(f)}{\operatorname{Lip}}\le (1-\kappa)\norm{f}{\operatorname{Lip}}
\end{align}
where $\norm{f}{\operatorname{Lip}}=\displaystyle \sup_{g\in G, h\in S}\sqrt{\omega(h)}\,|f(g)-f(gh)|$ is the weighted Lipschitz constant. Then for all $1\le p\le \infty$,
\[\norm{{\partial \cP(x)}}{\infty,p}\le (1-\kappa) \norm{{\partial x}}{\infty,p}\,.\]
\end{enumerate}
\end{theorem}
\begin{proof}
(i) is similar to Theorem \ref{transferrencetheorem}(i). (ii) For $x,A\in \cB(\cH)$, we denote the function $f_A^x(g):=\tr(u(g)xu(g^*)A)$.
Then,
\begin{align*} \norm{\partial x}{\infty,p}=\sup_{\|A\|_{p'}\le 1} 
\sup_{h\in S}\,|\Tr((\partial_h x)A)|
=&\sup_{\|A\|_{p'}\le 1}\sup_{h\in S}\,\left|\Tr\big( \left(\nabla_h\pi(x)(e)\right) A\big)\right|
\\ =& \sup_{\|A\|_{p'}\le 1}\,\sup_{h\in S}\sqrt{\omega(h)}\,\left|f_A^x(h)-f_A^x(e)\right|
\\ \overset{(1)}{=}& \sup_{\|A\|_{p'}\le 1}\,\sup_{g\in G,h\in S}\sqrt{\omega(h)}\,\left| f_A^x(gh)-f_A^x(g)\right|
\\ =&  \sup_{\|A\|_{p'}\le 1}\norm{ f_A^x}{\text{Lip}}\,.
\end{align*}
The identity (1) uses that for any $g\in G$ and $y\in \cB(\cH)$,
\begin{align} f_A^y(gh)-f_A^y(g)&=\Tr(Au(gh)yu(gh)^*)-\Tr(Au(g)yu(g)^*)\nonumber
\\&= \Tr(u(g)^*Au(g) u(h)yu(h)^*)-\Tr(u(g)^*Au(g) y)\nonumber
\\&=f_{u(g)^*Au(g)}^y(h)-f_{u(g)^*Au(g)}^y(e) \label{trickyequation1}\,.
\end{align}
Therefore, applying the commutation relation \eqref{eq:inter2}, we get
\begin{align*} \norm{\partial\cP(x)}{\infty,p}=\sup_{\|A\|_{p'}\le 1} 
\sup_{h\in S}\pl|\Tr(\partial_h\cP(x)A)|
&=\sup_{\|A\|_{p'}\le 1}\sup_{h\in S}\pl \pl \left|\Tr\big( \left(\nabla_h(K\pi(x))(e)\right) A\big)\right|
\\ &= \sup_{\|A\|_{p'}\le 1}\sup_{h\in S}\pl \left|\nabla_h Kf_A^x(e)\right|
\\ &\overset{(2)}{=} \sup_{\|A\|_{p'}\le 1}\,\sup_{g\in G,h\in S}\pl \sqrt{\omega(h)}\,\left|Kf_A^x(gh)-Kf_A^x(g)\right|
\\ &=  \sup_{\|A\|_{p'}\le 1}\norm{ Kf_A^x}{\text{Lip}}
\\ &\le  (1-\kappa)\sup_{\|A\|_{p'}\le 1}\norm{ f_A^x}{\text{Lip}}=\, \norm{\partial x}{\infty,p}\,,
\end{align*}
where (2) above follows again from \eqref{trickyequation1}. 
\end{proof}

\section{Examples and applications}\label{examples}

\subsection{Quantum Gibbs samplers}

Let $G=(V,E)$ be a hypergraph with $|V|=n$, and let $\cH_V:=\bigotimes_{v\in V}\cH_v$ be the Hilbert space of a local quantum system, where we assume that $\cH_v:=\mathbb{C}^d$ for some local dimension $d\in\mathbb{N}$. The interactions are modeled through the Hamiltonian $H:=\sum_{A\in E}h_A$ with local self-adjoint operators $h_A$ with $\|h_A\|\le 1$ supported on the hyperedges $A\in E$. Here, we also assume that the Hamiltonian is of finite-range, which means that the size and diameter of any hyperedge are uniformly bounded by a constant. The Gibbs state $\omega$ at inverse temperature $\beta>0$ is defined as
\begin{align}
    \omega:=\frac{e^{-\beta H}}{\tr\big[e^{-\beta H}\big]}\,.
\end{align}
A Gibbs sampler is a locally defined quantum channel which prepares an approximation of the Gibbs state $\omega $ starting from any initial state on $\cH_V$. The efficiency of the Gibbs sampler depends on the time it takes to reach the approximating state. In the recent years, various Gibbs samplers were proposed in the literature \cite{kastoryano2016quantum,Brando2018,Bardet2021,capel2020modified,mlsioned}. In \cite{Majewski1995,temme2015fast,de2021quantumb}, the authors prove curvature lower bounds for different Gibbs samplers and Lipschitz constants. Here, we provide a variant of these results which leads to a transportation information inequality with a tight scaling of the constant with the size $n$ of the system: for a given site $v\in V$, we denote the composition of the partial trace $\tr_v$ on $v$ with the Petz recovery map of $v$ as follows:
\begin{align}
    \Psi_v^\dagger (\rho)=\Phi_v^\dagger\circ \tr_v(\rho)=\omega^{\frac{1}{2}}\,(\omega_{v^c}^{-\frac{1}{2}}\rho_{v^c}\omega_{v^c}^{-\frac{1}{2}}\otimes I_v)\,\omega^{\frac{1}{2}}\,,
\end{align}
where $\omega$ is the Gibbs state of the Hamiltonian $H$, and  we denote by $\omega_A$ the reduced density on subregion $A\subseteq V$. Clearly, when $H$ is made of commuting terms, the map $\Psi_v$ acts non-trivially on the neighborhood of $v$, which is defined as
\begin{equation}
    N_v:=\bigcup\left\{A\in E : v\in A\right\}\,.
\end{equation}
Next, we introduce the generator of the heat-bath dynamics
\begin{align}
\cL_V:=\sum_{v\in V}\cL_v\,,
\end{align}
where $\cL_v:=\Psi_v^\dagger-\id$. The quantum Markov semigroup $t\mapsto e^{t\cL_V^\dagger}$ generated by $\cL_V^\dagger$ converges to $\omega$ as $t\to\infty$. Here, we prove a Lipschitz estimate for the semi-norm 
\[ \norm{x}{L}:=\max_{v\in V}\|x-\tau_v(x)\|\]
below some critical inverse temperature $\beta_c>0$, where $\tau_v(x)=\Tr_v(x)\ten \frac{1_v}{d}$ is the normalized partial trace at site $v$. This implies a coarse Ricci curvature lower bound independent of size $n$. Our analysis follows that of \cite{Majewski1995,temme2015fast} who instead considered the oscillator norm $\norm{x}{\operatorname{osc}}=\sum_{v\in V}\|x-\tau_v(x)\|$. 
\begin{proposition}
With the above notations and assuming $H$ is a sum of commuting terms $h_A$, there exists an inverse temperature $\beta_c>0$ such that for any $\beta<\beta_c$, there is a constant $\kappa(\beta)<1$ such that for all $x\in\cB(\cH_V)$,
\begin{align}
    \|e^{t\cL_V}(x)\|_L\le \,e^{-(1-\kappa(\beta))t}\,\|x\|_L\,.
\end{align}
\end{proposition}
\begin{proof}
For any $v\in V$, we denote $\partial_v:=\id-\tau_v$. Clearly, we have 
\begin{align*}
 \Psi_v(X)=\omega_{v^c}^{-1/2}\Tr_{v}(\omega^{1/2}X\omega^{1/2})\omega_{v^c}^{-1/2}\pl,\pl   \partial_v\circ \cL_v=\partial_v\circ \Psi_v- \partial_v=-\partial_v\,.
\end{align*}
Using this, we have
 \begin{align*}
 \frac{d}{ds}\,\partial_v\,e^{s\cL_V}=\partial_v\circ \cL_V e^{s\cL_V}=-\partial_v e^{s\cL_V}+\cL_{v^c}\circ  \partial_v e^{s\cL_V}   +\sum_{w\ne v}[\partial_v,\cL_w]\circ e^{s\cL_V}\,,
 \end{align*}
where we set $\cL_{v^c}=\sum_{w\ne v}\cL_w$. Therefore, 
\begin{align}
    \frac{d}{ds}\Big(e^{s}e^{(t-s)\cL_{v^c}}\partial_v e^{s\cL_V}\Big)=\sum_{w\ne v}\,e^{s} e^{(t-s)\cL_{v^c}}\big([\partial_v,\cL_{w}]e^{s\cL_V}\big)\,.
\end{align}
After integrating $s$ from $0$ to $t$ and using that $ e^{t\cL_{v^c}}$ is a contraction in operator norm, we find for any $x\in\cB(\cH_V)$
\begin{align}\label{eq:ineqmid}
 \|   \partial_v e^{t\cL_V}(x)\|\le e^{-t}\,\|\partial_v x\|+\,\int_0^t\,e^{s-t}\sum_{w\ne v}\|[\partial_v,\cL_w]e^{s\cL_V}(x)\|\,ds\,.
\end{align}
Moreover, since $\omega$ is the Gibbs state of a commuting Hamiltonian, $\cL_w$ is supported in a neighborhood $N_w$ of $w$ of size depending on the interaction range of $H$. Then, denoting the normalized partial trace corresponding to the region $N_w$ by $\tau_{N_w}$, we have that for any $x\in\cB(\cH_V)$ and $v\in N_w\backslash w$,
\begin{align*}
    \|[\partial_v,\cL_w](x)\|&=\|[\tau_v,\cL_w-\partial_w](x-\tau_{N_w}(x))\|\le 2\|\cL_w-\partial_w\|_{\operatorname{cb}}\|x-\tau_{N_w}(x)\|\,,
    \end{align*}
where the first identity comes from the fact that $[\tau_v,\cL_w-\partial_w]\tau_{N_w}(x)=0$. Moreover, $[\partial_v,\cL_w]=0$ whenever $v\in N_w^c$. Using these estimates in \Cref{eq:ineqmid}, we have that
\begin{align}
    \|   \partial_v e^{t\cL_V}(x)\|&\le e^{-t}\,\|\partial_v x\|+\,2\,\int_0^t\,e^{s-t}\sum_{w|N_w\ni v}\|\cL_w-\partial_w:L_\infty\to L_\infty\|_{\operatorname{cb}}\,\|(\id-\tau_{N_w})(e^{s\cL_V}(x))\|\,ds\nonumber\\
    &\le e^{-t}\,\| x\|_L +\,2\,\max_{v}n_v\max_w \|\cL_w-\partial_w:L_\infty\to L_\infty\|_{\operatorname{cb}}\,|N_w|
    \int_0^t\,e^{s-t}\|e^{s\cL_V}(x)\|_L \,ds\,,\label{eq:almost}
\end{align}
where $n_v$ denotes the number of vertices $w$ such that $v\in N_w$. Since $\cL_w\to\partial_w$ as the inverse temperature $\beta\to 0$, there exists a critical inverse temperature $\beta_c$ such that for any $\beta<\beta_c$, 
\begin{align*}
    \kappa(\beta):=2\,\max_{v}n_v\max_w \|\cL_w-\partial_w:L_\infty\to L_\infty\|_{\operatorname{cb}}\,|N_w|<1\,.
\end{align*}
Taking maximum over site $v$, we have 
\[\|e^{t\cL_V}(x)\|_L\le  e^{-t}\,\| x\|_L +\kappa(\beta)
    \int_0^t\,e^{s-t}\|e^{s\cL_V}(x)\|_L \,ds\,\pl.\]
Denote $f(t)=\|e^{t\cL_V}(x)\|_\otimes$. Differentiating the above inequality at $t=0$, we have 
\[ f'(0)\le -(1-\kappa(\beta))f(0)\pl.\]
By the semigroup property, $f'(t)\le -(1-\kappa(\beta))f(t)$ and by the Gronwall Lemma, this implies $f(t)\le e^{-(1-\kappa(\beta))t}f(0)$. That completes the proof.
\end{proof}

\subsection{Bosonic beam-splitter channels}
In this example, we consider Bosonic beam-splitter channels \cite{holevo2019quantum}. 
Recall that an $n$-mode Bosonic quantum system is modeled by the algebra $\cM:=\cB(L_2(\mathbb{R}^n))$. For all $i\in[n]$, we define the \textit{annihilation} operator $a_i$ as 
\begin{align}
    (a_i\psi)({\bf x})=\frac{x_i\,\psi({\bf x})+\partial_{x_i}\psi({\bf x})}{\sqrt{2}}\,, 
\end{align} and denote the \textit{creation} operator, i.e.  the adjoint of $a_i$, by $a_i^*$. Creation and annihilation operators satisfy the following \textit{canonical commutation relations}:
\begin{align}
    [a_i,a_j^*]=\delta_{ij}1\,,\qquad [a_i,a_j]=[a_i^*,a_j^*]=0\, .
\end{align}
Given a two-mode Bosonic system with creation operators $a_1\equiv a$ and $a_2\equiv b$, the   \textit{beam-splitter} of transmissivity $0\le \lambda\le 1$ is given by the unitary
\begin{align}
    U_\lambda=\exp\Big(\Big(a^*b-b^*a\Big)\,\arccos{\sqrt{\lambda}}\Big)\,
\end{align}
which performs linear rotations on the annihilation operators $a$ and $b$:
\begin{align}\label{eq1}
    U_\lambda^*aU_\lambda=\sqrt{\lambda}a+\sqrt{1-\lambda}b,\qquad  U_\lambda^*bU_\lambda =-\sqrt{1-\lambda}a+\sqrt{\lambda}b\,.
\end{align}
For a quantum state $\sigma$ on $L_2(\mathbb{R})$ and any $\lambda\ge 0$, we define the Bosonic beam-splitter channels 
\begin{align}\label{beamsplitter}
\cP_\lambda:\cB(L_2(\mathbb{R}))\to \cB(L_2(\mathbb{R})) \pl,\qquad 
    \cP_\lambda(x)=\tr_2\big[ (1\otimes \sigma) U_\lambda^*(x\otimes 1)U_\lambda\big]\,.
\end{align}
The state $\sigma$ is usually called the \textit{environment} state. Such channels were recently examined in the context of quantum information theory, where extensions of well-known information theoretic inequalities such as the entropy power inequality and information isoperimetric inequality were obtained (see \cite{DePalma2018} and the references therein for an up-to-date review of the topic). In around the same time, Carlen and Maas found the sharp entropic Ricci curvature bound for the quantum Ornstein-Uhlenbeck semigroup, which is given by the maps $\cP_{e^{-t}}$ with $\si$ being a thermal Gaussian state at a finite temperature \cite{carlen2017gradient}. 
Both the entropy power inequality and the entropic Ricci curvature lower bound were then shown to imply the (sharp) modified logarithmic Sobolev inequality for the aforementioned semigroup. In analogy with the commutative setting, these inequalities can be derived from an intertwining relation between the channel $\cP_\lambda$ and the derivations $[a,\cdot]$ and $[a^*,\cdot]$ which directly originates from \Cref{eq1}. Here, we make similar use of these relations in order to derive non-commutative coarse Ricci curvature bounds for the maps $\cP_\lambda$. Our differential structure is given by the derivations
\[ \partial=(\partial_a,\partial_{a^*})\pl ,\pl\quad  \partial_a(x)=[a,x]\pl ,\pl \quad  \partial_{a^*}(x)=[a^*,x]\pl . \]
Recall that a operator $x\in \cB(L_2(\mathbb{R}))$ is called a Schwartz operator if it can be written as $x=\int_{\mathbb{R}^2}f(z)e^{i(za+\bar{z}a^*)}dzd\bar{z}$ for some Schwartz function $f$ on $\mathbb{C}\cong\mathbb{R}^2$ \cite{keyl15}. We denote by $\mathcal{S}(L_2(\mathbb{R}))$ the set of Schwartz operator on $L_2(\mathbb{R})$. It is clear that $\mathcal{S}(L_2(\mathbb{R}))$ is a core for the derivation $\partial$ as a densely defined operator on $L_2(\mathbb{R})$. We consider the 
Lipschitz semi-norm
\[\norm{x}{L}:=\norm{\partial(x)}{}=\max\big\{\|[a,x]\|,\,\|[a^*,x]\|\big\}\,.\]





\begin{proposition}\label{prop:boseinter}
Let $0\le \lambda\le 1$ and $\cP_\la$ be the beam-splitter channel defined in \Cref{beamsplitter}. Then $\cP_\la$ satisfies the  intertwining relation $\partial\cP_\lambda=\sqrt{\lambda}(\cP_\lambda\ten \id_2)\partial$. Therefore,
\begin{enumerate}
\item[$\operatorname{(i)}$] for any Schwartz operator $x\in \mathcal{S}(L_2(\mathbb{R}))$,
\[ \norm{\partial\cP_\lambda (x)}\le\,\sqrt{ \lambda}\,\norm{\partial x}\,.\]
In particular, $(\cB(L_2(\mathbb{R})),\cB_L,\cP_\lambda)$ has coarse Ricci curvature lower bounded by $1-\lambda$.
\item[$\operatorname{(ii)}$] for any state $\rho$ and any Schwartz operator $x\in \mathcal{S}(L_2(\mathbb{R}))$,
\[ \norm{\partial\cP_\la (x)}{\rho}\le \lambda\,\norm{\partial x}{\cP_\la^{\dagger}(\rho)}\,.\]
In particular, $(\cB(L_2(\mathbb{R})),\cB_\partial,\cP_\lambda)$ has coarse Ricci curvature lower bounded by $1-\lambda$, where $\cB_\partial$ induces the quantum Wassersetin $2$-metric defined in Example \ref{exam:W2}.
\end{enumerate}
\end{proposition}

\begin{proof}
 For any $x\in \mathcal{S}(L_2(\mathbb{R}))$, 
\begin{align*}
    \big[a,\cP_\lambda(x)\big]=\cP_\lambda\big([a,x]\big)+ \tr_2\big[(1\otimes \sigma)\big[a,U_\lambda^*\big](x\otimes 1)U_\lambda\big]+\tr_2\big[(1\otimes \sigma)U_\lambda^*(x\otimes 1)\big[ a,U_\lambda\big]\big]\,.
\end{align*}
Then, it is easy to verify from \eqref{eq1} that
\begin{align}\label{eq:commutator}
    [a,U_\lambda]=\big(a(1-\sqrt{\lambda})+\sqrt{1-\lambda}b \big)U_\lambda\pl, \pl\qquad 
  [a,U_\lambda^*]=-U_\lambda^* \big(a(1-\sqrt{\lambda})+\sqrt{1-\lambda}b)\,.
\end{align}
Therefore, $\big[a,\cP_\lambda(x)\big]=\sqrt{\lambda}\cP_\lambda\big([a,x]\big)$. Since the same relations stand when replacing $a$ by $a^*$, we have the intertwining relation $$\,\partial(\cP_\la x)=\sqrt{\lambda}(\cP_\lambda\ten 1_2)\partial\,.$$
 The other assertions follow from Proposition \ref{prop:L1inter} and Theorem \ref{thm:L2inter}.
\end{proof}
We shall now compare $W_L$ with other transportation cost metrics in the literature. 
It was proved in \cite{rouze2019concentration} that (the statement is for finite dimensional systems but the proof works identically for the case here) 
\begin{align}\label{W1toW2}
    W_L(\rho_1,\rho_2)\le 2(\operatorname{cosh}(\beta/2))^{\frac{1}{2}}\,\mathcal{W}_{1}^{\operatorname{Bose},\beta}(\rho_1,\rho_2)\le 2(\operatorname{cosh}(\beta/2))^{\frac{1}{2}}\,\mathcal{W}_{2}^{\operatorname{Bose},\beta}(\rho_1,\rho_2)\,,
\end{align}
where $\mathcal{W}_1^{\operatorname{Bose},\beta}$ is the metric dual to the semi-norm
\begin{align}
    \vertiii{x}_{2,\beta}:=\sqrt{2\operatorname{cosh}(\beta/2)}\,\big(\|[a,x]\|^2+\|[a^*,x]\|^2\big)^{\frac{1}{2}}\,\le 2\sqrt{\cosh(\beta/2)} \,\norm{\partial(x)}{}\,.
\end{align}
The metric $\mathcal{W}_{2}^{\operatorname{Bose},\beta}$ is associated to the Bosonic Ornstein-Uhlenbeck semigroup given by the generator
\[\cL_\beta^{\operatorname{Bose}}(x)=\frac{1}{2}\Big(e^{\beta/2}(a^*[x,a]+[a^*,x]a)+e^{-\beta/2}(a[x,a^*]+[a,x]a^*) \Big)\pl.\]
Here $\beta$ is the inverse temperature and the semigroup $e^{-\cL_\beta^{\operatorname{Bose}} t}$ admits a unique invariant state
\begin{align}
    \sigma_\beta:=\frac{e^{-\beta a^*a}}{\tr\big(e^{-\beta a^*a}\big)}\,.
\end{align}
For a state $\rho$, we define the weighted multiplication operator acting on ${\bf V}=(V_1,V_2)\in \cB(\cH)\oplus \cB(\cH)$ as
\[\Lambda_{\rho,\beta}({\bf V}):=\Big(\int_{0}^1e^{-\beta(\frac{1}{2}-s)} \rho^s V_1\rho^{1-s}ds, \int_{0}^1 e^{\beta(\frac{1}{2}-s)} \rho^s V_2 \rho^{1-s}ds \Big).\]
where the pseudo-metric is defined as
\begin{align}
    \norm{x}{g,\gamma(s)}^2=\inf_{x=\partial^*(\Lambda_{\rho,\beta}{\bf V})}\lan {\bf V}, \Lambda_{\rho,\beta}{\bf V} \ran_{\tr}
\end{align}
with ${\bf V}$ that satisfies the continuity equation
$x=\partial^*(\Lambda_{\rho,\beta}{\bf V})$. The semigroup $e^{-\cL_\beta^{\operatorname{Bose}} t}$ is the gradient flow of relative entropy w.r.t $\mathcal{W}_{2}^{\operatorname{Bose},\beta}$, and it is proved in \cite[Theorem 8.6]{carlen2017gradient} that it
has entropic curvature bound $\sinh(\beta/2)$. The following transportation cost inequality then follows from \cite[Theorem 11.5]{Carlen2019} (see also \cite{datta2020relating,rouze2019concentration}), 
\begin{align}\label{TCW2}
    \mathcal{W}_2^{\operatorname{Bose},\beta}(\rho,\sigma_\beta)\le \sqrt{\frac{2D(\rho\|\sigma_\beta)}{\operatorname{sinh}(\beta/2)}}\,,
\end{align}
Combined with inequality \eqref{W1toW2}, we have the following Proposition. Recall that the energy of a state $\rho$ is defined as $E_\rho:=\tr(a^*a\rho )$.
\begin{proposition}\label{propfinalbla}Let $\beta>0$. The following bounds hold:
\begin{itemize}
    \item[$\operatorname{(i)}$]For any two states $\rho_1,\rho_2\in\cD(L_2(\mathbb{R}))$ with  $E:=\max\{E_{\rho_1},E_{\rho_2}\}$,
    \begin{align}\label{eq:triangle}W_L(\rho_1,\rho_2)\le \inf_{\beta>0} \sqrt{64\operatorname{coth}(\beta/2)\,\big(\beta E-\ln\big(1-e^{-\beta}\big)\big)}\,.\end{align}
\item[$\operatorname{(ii)}$] Let $\cP_\lambda$ be a beam-splitter channel with environment state $\sigma$ satisfying $E_\sigma<\infty$. Then, for $E:=\max\{E_\rho, (\sqrt{\lambda E_\rho} +\sqrt{(1-\lambda)E_\sigma})^2$\}, 
\begin{align}\label{eq:triangle2}\mathcal{J}_L(\rho):= W_L(\rho,\cP_\lambda^\dagger(\rho))\le\inf_{\beta>0} \sqrt{64\operatorname{coth}(\beta/2)\,\big(\beta E'-\ln\big(1-e^{-\beta}\big)\big)}\pl.\end{align}
\item[$\operatorname{(iii)}$] If in additional $[b,\si]$ is trace class, then $\cP_\lambda$ admits a unique invariant state $\rho_\infty$ such that $E_{\rho_\infty}<\infty$.
\end{itemize}
\end{proposition}
\begin{proof}
(i) By the triangle inequality of $W_2^{\text{Bose},\beta}$ and \eqref{TCW2}, we find
\begin{align}\label{eq:triangle111}W_L(\rho_1,\rho_2)\le \sqrt{\frac{8\cosh(\beta/2)}{\sinh(\beta/2)}}\Big(\sqrt{D(\rho_1\|\sigma_{\beta})}+\sqrt{D(\rho_2\|\sigma_{\beta})}\Big)\,.\end{align}
Moreover, evaluating the relative entropy we have
\begin{align*}
D(\rho\|\sigma_\beta)&=-H(\rho) -\tr(\rho\ln\sigma_\beta)
\\ & \le -\tr(\rho\ln\sigma_\beta)
\\&= \beta\tr(\rho a^*a)+\tr(e^{-\beta a^*a})=\beta E-\ln (1-e^{-\beta})\,.
\end{align*}
This prove (i).
For (ii), it suffices to evaluate the energy of $\cP_\lambda^\dagger(\rho)$ in terms of that of $\rho$ and that of $\sigma$:
\begin{align}
    \tr\big( \cP_\lambda^\dagger(\rho)a^*a\big)=&
    \tr\big( \rho\cP_\lambda(a^*a)\big)
    =\tr\big( (\rho\ten \si) U_\lambda^*(a^*a)U_\lambda\big)
  \nonumber  \\=&\tr\Big(\rho\ten \si\Big( (U_\lambda^*a^*U_\lambda) (U_\lambda^*aU_\lambda)\Big)\Big)
  \nonumber  \\=&\tr\Big(\rho\ten \si\Big( (\sqrt{\lambda} a+\sqrt{1-\lambda}b)^* (\sqrt{\lambda} a+\sqrt{1-\lambda}b)\Big)\Big)
  \nonumber \\=&\tr\Big(\rho\ten \si\Big(\lambda a^*a+ \sqrt{(1-\lambda)\lambda} (a^*b+b^*a)+(1-\lambda) b^*b \Big) \Big)
  \nonumber \\ \le& \lambda E_\rho +(1-\lambda)E_\sigma +2\sqrt{(1-\lambda)\lambda E_\rho E_\sigma}
   \nonumber\\ \le& (\sqrt{\lambda E_\rho} +\sqrt{(1-\lambda)E_\sigma})^2\,.\label{eq:estimate}
\end{align}
The result follows from (i). For (iii), we choose any $\rho$ such that $E_\rho<\infty$ and write $\rho_n=(\cP_\la^\dagger)^n(\rho)$. We have by Proposition \ref{prop:boseinter} that 
\begin{align*}
W_L(\rho_n,\rho_{n+1})\le \lambda^{\frac{n}{2}} W_L(\rho,\rho_{1})= \lambda^{\frac{n}{2}}\mathcal{J}_L(\rho)\,.
\end{align*}
Combined with Proposition \ref{regularity} below, this implies for $n>1$
\begin{align}\norm{\rho_n-\rho_{n+1}}{1}\le\lambda^{\frac{n-1}{2}}\sqrt{\frac{\la}{1-\la}}\,\norm{[b,\sigma]}{1}W_{L}(\rho,\rho_1)\pl. \label{eq:L1est}\end{align}
Thus $\{\rho_n\}_{n\ge 0}$ is a Cauchy sequence in $\cD(
L_2(\mathbb{R}))$, which converges to some limit $\rho_\infty$ in $\norm{\cdot}{1}$.
 $\rho_\infty$ is an invariant state of $\cP_\lambda^\dagger$ because
\[ \norm{\cP_\la^\dagger(\rho_\infty)-\rho_\infty}{1}=\lim_{n\to \infty}\norm{\rho_{n+1}-\rho_n }{1}=0\,.\]
Moreover by the contraction of $W_L$, such an invariant state $\rho_\infty$ is unique. Indeed, if there is another invariant state $\omega_\infty$, we have for any $n$
\begin{align*} \norm{\rho_\infty-\omega_\infty}{1}&=\norm{(\cP_\la^\dagger)^{n}(\rho_\infty)-(\cP_\la^\dagger)^{n}(\omega_\infty)}{1}\\& \le \la^{\frac{n-1}{2}}\sqrt{\frac{\la}{1-\la}}\,\norm{[b,\sigma]}{1} W_L(\rho_\infty,\omega_\infty)\\ &\le 
 \la^{\frac{n-1}{2}}\sqrt{\frac{\la}{1-{\la}}}\, (W_L(\rho_\infty,\sigma_\beta)+W_L(\omega_\infty,\sigma_\beta))\pl.
\end{align*}Note that by assumption on $\rho$, we have for any $\beta>0$
\begin{align*}W_L(\rho_\infty,\sigma_\beta)&\le W_L(\rho_\infty,\rho )+W_L(\rho,\sigma_\beta )\\ &\le \frac{1}{1-\sqrt{\lambda}}\,\mathcal{J}_L(\rho)+\sqrt{\frac{8\cosh(\beta/2)D(\rho\|\si_\beta)}{\sinh(\beta/2)}}<\infty \end{align*}
Similarly, $W_L(\omega_\infty,\si_\beta)<\infty$, hence $W_L(\rho_\infty,\omega_\infty)<\infty$ and this implies
$\norm{\rho_\infty-\omega_\infty}{1}=0$. 
To see that $\rho_\infty$ has finite energy, we denote by $E_n=E_{(\cP_\lambda^\dagger)^n\rho}$. By the estimate \eqref{eq:estimate}, 
\[\sqrt{E_{n+1}}\le  \sqrt{\lambda E_{n}}+\sqrt{(1-\lambda)E_\sigma}\,.\]
Solving this inequality, we have
\[ \sqrt{E_{n}}\le \sqrt{\lambda}^n(\sqrt{E_0}-\frac{1-\sqrt{\lambda}}{\sqrt{1-\la}}\sqrt{E_{\sigma}})+\frac{1-\sqrt{\lambda}}{\sqrt{1-\la}}\sqrt{E_{\sigma}}\pl. \]
This implies   $\displaystyle E_{\rho_\infty}\le \lim_{n\to \infty}E_{n}\le \frac{(1-\sqrt{\la})^2}{1-\la}E_\sigma.$
\end{proof}
Finally, we prove the inequality \eqref{eq:L1est} used above and which upper bounds the trace distance between two output states of $\cP_\lambda^\dagger$ in terms of the Wasserstein distance $W_{L}$. Classically, this amounts to proving the following regularity in the dual picture:
\begin{align}
    \|\partial f\ast_\lambda h\|_{L_\infty(\mathbb{R})}&=\sup_{t\in\mathbb{R}}\Big|\partial_t\int\, f\big(\sqrt{\lambda}\,t+\sqrt{1-\lambda}\,s\big)\,h(s)\,ds
    \,\Big|\\
    &=\frac{1}{\sqrt{1-\lambda}}\sup_{t\in\mathbb{R}}\Big|\partial_t\int\, f(u)\,h\Big(\frac{u-\sqrt{\lambda}t}{\sqrt{1-\lambda}}\Big)\,du
    \,\Big|\\
    &\le \sqrt{\frac{\lambda}{1-\lambda}}\,\|f\|_{L_\infty(\mathbb{R})}\,\|\partial h\|_{L_1(\mathbb{R})}\,.
\end{align}
where $\ast_\lambda$ is the convolution with parameter $\lambda$. Then by duality, 
\begin{align}
    \|\mu_1\ast_\lambda h-\mu_2\ast_\lambda h\|_{\operatorname{TV}}\le \|\partial h\|_{L_1(\mathbb{R})}\sqrt{\frac{\lambda}{1-\lambda}}W_{L}(\mu_1,\mu_2)\,.
\end{align}
The next proposition extends this idea to the quantum setting:
\begin{proposition}\label{regularity}
For any two states $\rho_1,\rho_2\in\cD(L_2(\mathbb{R}))$, we have
\begin{align}
    \|\cP_{\lambda}^\dagger(\rho_1-\rho_2)\|_1\le \sqrt{\frac{\lambda}{1-\lambda}}\,\|[b,\sigma]\|_1\, W_{L}(\rho_1,\rho_2)\,.
\end{align}
\end{proposition}
\begin{proof}
For any Schwartz operator $x\in\mathcal{S}(L_2(\mathbb{R}))$, we have
\begin{align}
    [a,\cP_\lambda(x)]&=\big[a,\tr_2\big[(1\otimes \sigma)U_\lambda^*(x\otimes 1)U_\lambda\big]\big]\\
    &=\big[a,\tr_2\big[(1\otimes \sigma)U_{1-\lambda}(1\otimes x)U_{1-\lambda}^*\big]\big]\\
    &\overset{(1)}{=}(1-\sqrt{1-\lambda})[a,\cP_\lambda(x)]+\sqrt{\lambda}\tr_2(1\otimes\sigma)\big[b,U_{1-\lambda}(1\otimes x)U^*_{1-\lambda}\big]\\
    &\overset{(2)}{=}(1-\sqrt{1-\lambda})[a,\cP_\lambda(x)]-\sqrt{\lambda}\tr_2\Big[(1\otimes[b,\sigma])U_{1-\lambda}(1\otimes x)U_{1-\lambda}^*\Big]\,.
\end{align}
where (1) uses the the commutation relation \eqref{eq:commutator}, and (2) uses the tracial property in \cite[Theorem 17]{brown1990jensen}.
Therefore
\begin{align}
    [a,\cP_\lambda(x)]=-\sqrt{\frac{\lambda}{1-\lambda}}\,\tr_2\Big[(1\otimes [b,\sigma])U_{1-\lambda}(1\otimes x)U_{1-\lambda}^*\Big]\,.
\end{align}
Taking the operator norm:
\begin{align}
   \| [a,\cP_\lambda(x)]\|&=\sqrt{\frac{\lambda}{1-\lambda}}\,\sup_{\|y\|_1\le 1}\tr \big[(y\otimes [b,\sigma])\,U_{1-\lambda}(1\otimes x)U_{1-\lambda}^*\big]\\
   &\le \sqrt{\frac{\lambda}{1-\lambda}}\,\|[b,\sigma]\|_1\,\|x\|\,,
\end{align}
and similarly for $\|[a^*,\cP_\lambda(x)]\|$. The result follows by duality.  
\end{proof}
As a corollary, we obtain mixing times for the beam-splitter channel $\cP_\lambda$:
\begin{corollary}
For any $\rho_1,\rho_2\in\cD(L_2(\mathbb{R}))$ with finite energy,
and any $n\in\mathbb{N}$, 
\begin{align}
    \|(\cP_{\lambda}^\dagger)^n(\rho_1-\rho_2)\|_1\le \frac{16\,\lambda^{\frac{n}{2}}}{\sqrt{1-\lambda}\,(1-\sqrt{\lambda})}\,\|[b,\sigma]\|_1\,\inf_{\beta>0} \sqrt{\operatorname{coth}(\beta/2)\,\big(\beta E-\ln\big(1-e^{-\beta}\big)\big)}\,
\end{align}
where \[E:=
(\sqrt{\lambda \max\{E_{\rho_1},E_{\rho_2}\}} +\sqrt{(1-\lambda)E_\sigma})^2\] 
\end{corollary}
\begin{proof}
This is a direct application of Lemma \ref{distancetojump}, Propositions \ref{regularity}, \ref{prop:boseinter} and \ref{propfinalbla}.
\end{proof}

\subsection{Fermionic systems}
The curvature of a beam-splitter channel can also be obtained in the context of Fermionic systems, which we briefly outline in this subsection.
Recall that an $n$-mode Fermionic system coincides with an $n$-qubit system which can be described by the Clifford generators $\{c_1,\cdots ,c_{2n}\}$ satisfying the CAR relation
\[c_i=c_i^*, \qquad c_{i}c_j+c_{j}c_{i}=2\delta_{i,j}1\pl.\]
For each $1\le j\le n$, the annihilation and creation operators are given by
\[a_j=\frac{1}{2}(c_{2j-1}+ic_{2j})\,,\qquad  a_j^*=\frac{1}{2}(c_{2j-1}-ic_{2j})\pl.\]
Each element $x\in \mathbb{M}_{2^n}(\mathbb{C})$ can be written as
\[x=\sum_{A\subset [n] }\la_{A}c_A\,.\]
Here, $A$ is a subset of $[n]=\{1,\cdots, n\}$ and $c_A=c_{i_1}\cdots c_{i_k}$ is the (ordered) product for $A=\{i_1,\cdots, i_k\}$ and $i_1<i_2<\cdots<i_k$. In particular, $x$ is an even element if $x$ only contains elements $C_A$ with support $|A|$ even. 

We denote the algebra as $\mathcal{C}\cong\mathbb{M}_{2^n}(\mathbb{C})$. Let $\hat{\mathcal{C}}$ be a copy of $\mathcal{C}$ with generator $\hat{c}_1,\cdots,\hat{c}_n$. As in the Bosonic case, there is a \textit{beam-splitter} unitary $U_\lambda\in \mathcal{C}\ten \hat{\mathcal{C}}$ of transmissivity $0\le \lambda\le 1$ (see \cite[Lemma 3.1]{lust1999riesz}), 
which performs linear rotations on the Clifford generators 
\begin{align}\label{eq1f}
    U_\lambda^* c_j U_\lambda=\sqrt{\lambda}c_j+\sqrt{1-\lambda}\hat{c}_j\pl ,\pl \qquad   U_\lambda^*\hat{c}_jU_\lambda =-\sqrt{1-\lambda}c_j+\sqrt{\lambda}\hat{c}_j\,.
\end{align}
Let $\sigma$ be a state on $\hat{\mathcal{C}}$. For any 
$0\le \lambda
\le 1$, we define the quantum channel 
\begin{align}
\cP_\lambda:\mathbb{M}_{2^n}(\mathbb{C})\to \mathbb{M}_{2^n}(\mathbb{C}) \pl,\pl\qquad 
    \cP_\lambda(x)= \tr_2\Big(( 1\ten \sigma)\, U_\lambda^*(x\ten 1) U_\lambda\Big)\,.
\end{align}
where $\tr_2$ is the trace on $\hat{\mathcal{C}}$. We
consider the differential structure given by the derivation
\[\partial=(\partial_{j})_{j=1}^{2n}:\mathcal{C}\to \bigoplus_{j=1}^{2n}\mathcal{C}\pl, \pl\qquad  \partial_{j}(x)=[c_j,x].\]
We have the following analog of Proposition \ref{prop:boseinter}
\begin{proposition}\label{prop:ferminter}
Let $0\le \lambda\le 1$, $\cP_\la$ be the Bosonic channel defined above. Then $\cP_\la$ satisfies the intertwining relation $\partial\cP=\sqrt{\lambda}\,\overset{\rightarrow}{\cP}\partial$ with $\overset{\rightarrow}{\cP}=(\cP, \cP,\dots, \cP)$. Moreover,
\begin{enumerate}
\item[$\operatorname{(i)}$]Define the semi-norm
$\norm{x}{L}=\max_{j}\norm{\partial_{j}x}{}$. Then for any $x\in \mathbb{M}_{2^n}(\mathbb{C})$
\[ \norm{\cP_\lambda x}{L}\le\sqrt{ \lambda}\norm{ x}{L}\,.\]
Therefore, $(\mathbb{M}_{2^n}(\mathbb{C}),\cB_L,\cP_\lambda)$ has coarse Ricci curvature lower bounded by $1-\sqrt{\lambda}$.
\item[$\operatorname{(ii)}$] for any state $\rho$ and operator $x\in \mathbb{M}_{2^n}$,
\[ \norm{\partial\cP_\la x}{\rho}\le \lambda\norm{\partial x}{\cP_\la^{\dagger}\rho}\,.\]
where $\norm{\cdot}{\rho}$ is the weighted norm defined in \eqref{eq:weight}.
Therefore, for the binary relation $\cB_\partial$ \eqref{CarlenMaas} inducing the Wasserstein $2$-metric, $(\mathbb{M}_{2^n}(\mathbb{C}),\cB_\partial,\cP_\lambda)$ has coarse Ricci curvature lower bounded by $1-\lambda$. 
\end{enumerate}
\end{proposition}

\begin{proof}
The proof is identical to  Proposition \ref{prop:boseinter}. 
\end{proof}
\begin{rem}
{\rm The channel $\cP_\lambda$ is different from the Gaussian channel considered in \cite{bravyi2004lagrangian}. The latter uses a \textit{beam-splitter} unitary $U_\lambda$ satisfying \eqref{eq1f} for $c_1,\cdots, c_{2n}$ and $\hat{c}_1,\cdots, \hat{c}_{2n}$
forming a $2n$-mode Fermionic system, as opposed to the tensor product $\mathcal{C}\ten \hat{\mathcal{C}}$ considered here. For such a channel, we can only obtain the  estimates in Proposition \ref{prop:ferminter} for even elements.
}
\end{rem}

\subsection{Pauli channels}
Recall that the Pauli matrices in $\cM=\mathbb{M}_2(\mathbb{C})$ are
\[\si_0=1\pl ,\pl  \si_1=\left[\begin{array}
{cc}     0&1  \\
     1& 0
\end{array}\right]\pl ,\pl  \si_2=\left[\begin{array}
{cc}     0&-i  \\
     i& 0
\end{array}\right]\pl , \pl  \si_3=\left[\begin{array}
{cc}     1&0 \\
     0& -1
\end{array}\right] \pl.\]
We consider the Pauli channel on the $n$-qubit system \[
\cP:\mathbb{M}_{2^n}(\mathbb{C})\to \mathbb{M}_{2^n}(\mathbb{C})\pl, \pl\qquad 
\cP(x):=\sum_{\alpha\in \mathcal{I}}\,\lambda_\alpha \sigma_\alpha x\,\sigma_\alpha\,, \] where $\mathcal{I}\subset \{0,1,2,3\}^n$ is a multi-index subset and given a string $\alpha\in\{0,1,2,3\}^n$,  $\sigma_\alpha=\bigotimes_{i=1}^n\sigma_{\alpha_i}$ is the tensor product of Pauli matrices. Note that each  $\sigma_\alpha$ is a self-adjoint unitary and
\[ \si_\al\si_{\beta}=(-1)^{c(\al,\beta)}\si_{\beta}\si_\al\pl \]
commute up to a sign depending on $\al$ and $\beta$. The coefficients $\lambda_\alpha$ form a probability distribution over the set $\mathcal{I}$. We assume that $\lambda_\alpha>0$ for any $\al\in\mathcal{I}$ and ${\bf 0}\in\mathcal{I} $. We consider the differential structure given by derivation
\[\partial: \cM\to \bigoplus_{\al\in \mathcal{I}}\cM\pl, \pl\qquad  \partial_{\al}(x)=[\si_\al,x]\pl.\]
For each $\beta\in \mathcal{I}$, we have $\partial_\beta(\si_\beta x\si_\beta)=-\partial_\beta(x)$ and
\begin{align*}\partial_\beta(\si_\al x\si_\al)=\si_\beta\si_\al x\si_\al-\si_\al x\si_\al\si_\beta=(-1)^{c(\al,\beta)}( \si_\al\si_\beta x\si_\al-\si_\al x\si_\beta\si_\al)
=(-1)^{c(\al,\beta)}\si_\al\partial_\beta(x)\si_\al\,.
\end{align*}
Then 
\begin{align*}
   \partial_\beta\circ\cP(x)&=\sum_{\alpha\in\mathcal{I}}\,\lambda_\alpha\,\partial_\beta(\sigma_\alpha\,x\,\sigma_\alpha)\\
    &=\sum_{\substack{\alpha\in\mathcal{I}\\\alpha\ne\beta}}(-1)^{c(\al,\beta)}\lambda_\alpha\,\sigma_\alpha\,\partial_{\beta}x\,\si_\al-\lambda_\beta\delta_\beta(x)\,\\
    &=(\lambda_\mathbf{0}-\lambda_\beta)\partial_\beta(x)\,+\sum_{\substack{\alpha\in\mathcal{I}\\\alpha\notin\{\mathbf{0},\beta\}}}\,(-1)^{c(\al,\beta)}\lambda_\alpha \sigma_\alpha\partial_\beta(x)\,\sigma_\alpha\,.
\end{align*}
Denote $\la(\beta)=\min\{\la_{\bf 0},\la_\beta \} $.
We define the map $$\overrightarrow{\cP}=(\cP^{(\beta)})_{\beta\in \mathcal{I}}\pl,\qquad  \pl  \cP^{(\beta)}(x):=(1-2\lambda(\beta))^{-1}\Big((\lambda_{\mathbf{0}}-\lambda_\beta)x+\sum_{\substack{\alpha\in\mathcal{I}\\\alpha\notin\{\mathbf{0},\beta\}}}(-1)^{c(\al,\beta)}\lambda_\alpha\sigma_\alpha\,x\,\sigma_\alpha\Big)\,.$$
Then we have $\partial\cP= \cP^{(\beta)}\partial$. Note that for each $\beta$, $\cP^{(\beta)}$ is a complete contraction on $L_p(\cM)$ for $1\le p\le \infty$ because
\[|\la_0-\la_\beta|+\sum_{\al\in \mathcal{I}/\{{\bf 0},\beta\}}\la_\alpha=1-2\lambda(\beta)\,.
\]
Define the semi-norm
\[\norm{x}{\partial,\infty}=\max_{\al\in \mathcal{I}}\norm{\partial_{\beta}x}{}\,. \]
Therefore, 
\begin{align}
    \|\cP(x)\|_{\partial,\infty}&=\max_{\beta\in \mathcal{I}}  \big(1-2\lambda(\beta)\big)\,\|\cP^{(\beta)}\partial_\beta(x)\|\\ &\le \max_{\beta\in \mathcal{I}}   \big(1-2\lambda(\beta)\big)\,\|\partial_\beta(x)\|\\ &\le (1-2\min_{\beta\in\mathcal{I}}\lambda(\beta))\norm{ x}{\partial,\infty}
    \\ &=(1-2\min_{\beta\in\mathcal{I}}\lambda_\beta)\norm{ x}{\partial,\infty}\,.
\end{align}
In other words, the triple $(\mathbb{M}_{2^n}(\mathbb{C}),\cB_{\partial,\infty},\cP)$ has coarse Ricci curvature lower bounded by $2\min_{\beta\in\mathcal{I}}\lambda_\beta$\,.
Moreover, by the triangle inequality 
\begin{align*}
  \,\norm{ x}{\partial,\infty}\le  2\|x\|\,.
\end{align*}
Then, we have by duality that for any two states $\rho,\sigma$, $\frac{1}{2}\|\rho-\sigma\|_1\le W_\Gamma(\rho,\sigma)$. 
Denote $$E_{\mathcal{I}}(\rho)=\frac{1}{|\mathcal{I}'|}\sum_{\gamma \in\mathcal{I}'}\sigma_\gamma x \sigma_\gamma, $$ where $\mathcal{I}'$ is the set of Pauli strings generated by those in $\mathcal{I}$. It is easy to see that $\cP^k(\rho)\to E_{\mathcal{I}}(\rho)$ for any state. Denote $W_{\partial,\infty}$ to be the Wasserstein $1$-distance dual to $\norm{\cdot}{\partial,\infty}$
Then we conclude that for all $k$,
\begin{align}
    \|\cP^{k}(\rho)-E_{\mathcal{I}}(\rho)\|_1&\le 2W_{\partial,\infty}(\cP^{k}(\rho),E_{\mathcal{I}}(\rho))\le\,\big(1-2\min_{\beta\in\mathcal{I}}\lambda_\beta\big)^k W_{\partial,\infty}(\rho, E_{\mathcal{I}}(\rho))\\ &\le
    \frac{\big(1-2\min_{\beta\in\mathcal{I}}\lambda_\beta\big)^k}{2\min_{\beta\in\mathcal{I}}{\lambda_\beta}} \,\mathcal{J}_{\partial,\infty}(\rho)
    \,.
\end{align}
where $\mathcal{J}_{\partial,\infty}(\rho):=W_{\partial,\infty}(\rho,\cP(\rho))$ is the jump of $\rho$. We summarise the above discussion in the following proposition.
\begin{proposition}
Let $\cP(x):=\sum_{\alpha\in \mathcal{I}}\,\lambda_\alpha \sigma_\alpha x\,\sigma_\alpha$ be a Pauli channel. Suppose ${\bf 0}\in \mathcal{I}$ and $\la_\al>0$ for any $\al \in \mathcal{I}$. Then $\cP$ has coarse Ricci curvature lower bound $2\min_{\beta\in\mathcal{I}}$ to the metric $W_{\partial,\infty}(\rho,\cP(\rho))$ defined above. Moreover, for any state $\rho$ and $k\ge 1$,
\[ \|\cP^{k}(\rho)-E_{\mathcal{I}}(\rho)\|_1\le \frac{\big(1-2\min_{\beta\in\mathcal{I}}\lambda_\beta\big)^k}{2\min_{\beta\in\mathcal{I}}{\lambda_\beta}} \,\mathcal{J}_{\partial,\infty}(\rho)\,.\]
\end{proposition}



  





\bibliographystyle{abbrv}
\bibliography{references}

\end{document}